\documentclass[10pt,twocolumn,twoside]{IEEEtran}
\IEEEoverridecommandlockouts 
\usepackage{bm}
\usepackage{amssymb,amsmath,graphicx}
\usepackage{float}
\usepackage{color,soul}
\usepackage{booktabs}
\usepackage{wrapfig}
\usepackage{mathtools}
\usepackage{algorithm, algpseudocode}
\newtheorem{theorem}{Theorem}

\newtheorem{definition}{Definition}
 
\newtheorem{proposition}{Proposition}

\newcommand{\norm}[1]{\left\lVert#1\right\rVert}

\newcommand{\argmin}{\mathrm{arg}\min}

\begin{document}

\title{Controller Synthesis of Wind Turbine Generator and Energy Storage System with Stochastic Wind Variations under Temporal Logic Specifications
}
\author{Zhe~Xu, Agung~Julius, Ufuk Topcu and Joe H. Chow
	}

\maketitle

\begin{abstract}
In this paper, we present a controller synthesis approach for wind turbine generators (WTG) and energy storage systems with metric temporal logic (MTL) specifications, with provable probabilistic guarantees in the stochastic environment of wind power generation. The MTL specifications are requirements for the grid frequency deviations, WTG rotor speed variations and the power flow constraints at different lines. We present the stochastic control bisimulation function, which bounds the divergence of the trajectories of a switched stochastic control system and the switched nominal control system in a probabilistic fashion. We first design a feedforward controller by solving an optimization problem for the nominal trajectory of the deterministic control system with robustness against initial state variations and stochastic uncertainties. Then we generate a feedback control law from the data of the simulated trajectories. We implement our control method on both a four-bus system and a nine-bus system, and test the effectiveness of the method with a generation loss disturbance. We also test the advantage of the feedback controller over the feedforward controller when unexpected disturbance occurs. 			
\end{abstract}  


\IEEEpeerreviewmaketitle
	\section{Introduction} \label{Introduction}
With the increasing penetration of renewable energy into the power grid, the stochastic nature of the renewable energy such as wind energy makes the planning and operation of power grid more challenging. Energy storage systems such as battery energy storage systems can be utilized to compensate the volatility brought by the renewable energy and regulate the grid frequency within the allowable range~\cite{zhe2017, Mondal2015, zheACC}. On the other hand, the renewable energy can also be utilized for the frequency regulation when disturbances such as generation loss or line failure occurs. For example, it has been shown that wind turbine generators can adjust its power output for restoring the grid frequency after a disturbance~\cite{ZhangPulgar2017}.
	
As different generators and energy storage systems have different response time, the regulated
frequency could have different temporal properties. Therefore, temporal logics can be utilized to provide time-related specifications such as ``after a disturbance, the grid frequency should be restored to 60Hz$\pm$0.5Hz within 2 seconds and to 60Hz$\pm$0.3Hz within 20 seconds''. There is much research on the control of wind turbine generators or energy
storage systems for economic or stability purposes, while incorporating temporal logic constraints into the controller
synthesis problem is still a novel approach. 

\subsection{Related Works} 	
There have been various methods on how to design
controllers to meet temporal logic specifications in deterministic environment \cite{zhe2017, Zhe_intermittent, zhe_advisory} and stochastic environment
\cite{Nonlinear2017,Wolff2012}. For discrete-time temporal logic specifications such as linear temporal logic (LTL) specifications, the usual approach is to abstract the system as a Markov Decision Process (MDP), then the control design is transformed to a problem of finding the control strategy that maximizes the probability of producing a sequence of states in the MDP satisfying the LTL specification \cite{Lahijanian2012}. For dense-time temporal logic specifications, the system can be abstracted as a timed automaton \cite{AlurDill1990,Fu2015CDC,zhe_ijcai2019} and the design process reduces to reachability analysis after the timed automaton is constructed. 

\subsection{Contributions and Advantages} 	

\noindent\textbf{1. Feedforward controller synthesis for wind turbine generator and energy storage systems with stochastic wind variations and uncertainties in the initial state with respect to MTL specifications:} In this paper, we present a controller synthesis approach for metric temporal logic specifications in stochastic environment by designing the controller for 
the trajectory of the diffusionless version of the process with robustness against initial state variations and stochastic uncertainties. Our work is motivated by the works of probabilistic testing for stochastic systems in \cite{Julius2008CDC,Prajna2004}. We present the stochastic control bisimulation function, which bounds the divergence of the trajectories of the stochastic control system and the nominal deterministic control system in a probabilistic fashion. Thus all the controller synthesis methods for the deterministic system can be used for designing the optimal input signals, and the same input signals can be applied to the stochastic system with a probabilistic guarantee. A preliminary version of this paper appeared in conference proceedings \cite{zheACC2018wind} where we present a coordinated control method of wind turbine generator and energy storage system for frequency regulation under temporal logic
specifications. In this paper, we extend the results to switched stochastic control systems, and generalized the predicates of the MTL specifications to include both the state and the input so that line power constraints can be incorporated into the MTL specifications. Besides, we add an exponential term to the (stochastic) autobisimulation function so that both stable and unstable linear dynamics can be approached with less conservativeness.

\noindent\textbf{2. Feedback controller synthesis for power systems with respect to MTL specifications for wind turbine generator and energy storage systems with stochastic wind variations:} To account for unexpected disturbances, we generate a feedback control law from the data of the simulated trajectories to form a feedback controller. We apply the controller synthesis method in regulating the grid frequency of both a four-bus system and a nine-bus system by the coordinated control of wind turbine generators and energy storage systems. We test the effectiveness of the controllers with the stochastic wind generation after a large generation loss disturbance, and also the case when unexpected disturbances are added to the situation. Simulations have shown that the trajectories with the feedforward and the feedback controllers can satisfy the MTL specification with a probabilistic guarantee. Besides, simulations show that when even unexpected wind gust disturbances occur, the feedback controller has better performance in comparison with the feedforward controller.

	\section{Preliminaries} 
	\label{Preliminaries}  
\subsection{Switched Stochastic Control Systems}
\label{switch} 
\begin{definition}[Switched Stochastic Control Systems]
A \textbf{switched stochastic control system} is a 6-tuple $\mathcal{T} = (\mathcal{Q},\mathcal{X},\mathcal{X}_0,\mathcal{V},\mathcal{F},\mathcal{E})$ where
\begin{itemize}
	\item $\mathcal{Q}=\{1,2,\dots,M\}$ is the
	set of indices for the modes (or subsystems);
	\item $\mathcal{X}$ is the domain of the continuous state, $x\in \mathcal{X}$ is the continuous state of the system, $\mathcal{X}_0\subset\mathcal{X}$ is the initial set of states;	
	\item $\mathcal{V}$ is the domain of the input, $u\in\mathcal{V}$ is the input of the system;		
	\item $\mathcal{F}=\{(f_{q},g_{q})\vert q\in \mathcal{Q}\}$ where $f_{q}$ describes the continuous
	time-invariant dynamics for the mode $dx=f_{q}(x,u)dt+g_{q}(x,u)dw$, which
	admits a unique solution $\xi_{q}(t;x_{q}^0,u)$, where $\xi_{q
	}$ satisfies $d\xi_{q}(t;x_{q}^0,u)= f_{q}(\xi_{q}(t;x_{q}^0,u),u)dt+g_{q}(\xi_{q}(t;x_{q}^0,u),u)dw$, and $\xi_{q}(0;x_{q}^0,u)=x_{q}^0$ is an initial condition in mode $q$;
	\item $\mathcal{E}$ is a subset of $\mathcal{Q}\times \mathcal{Q}$
	which contains the valid transitions. If a transition $e = (q, q')\in\mathcal{E}$ takes place, the system switches from mode $q$ to $q'$.
\end{itemize}
\label{sw}
\end{definition}                      
Similarly, we can define the \textbf{switched nominal control system} $\mathcal{T}^{\ast} = (\mathcal{Q},\mathcal{X},\mathcal{X}_0,\mathcal{V},\mathcal{F}^{\ast},\mathcal{E})$ of $\mathcal{T}$, and $\mathcal{T}^{\ast}$  only differs from $\mathcal{T}$ as $\mathcal{F}^{\ast}=\{f_{q}\vert q\in \mathcal{Q}\}$, where $dx^{\ast}=f_{q}(x^{\ast},u)dt$ is the nominal deterministic version of $dx=f_{q}(x,u)dt+g_{q}(x,u)dw$.

 We assume that there is a minimal dwell time $\Lambda_{\rm{min}}$ between any two transitions to avoid too frequent transitions. 
  
\begin{definition}[Trajectory]
	\label{def_traj}	
	A trajectory of a stochastic switched control system $\mathcal{T}$ is denoted as a sequence $\rho=\{(q^{i},\xi_{q^{i}}(t;x^{0}_{q^i},u),T^{i})\}_{i=0}^{N_{q}}$ ($N_{q}\in\mathbb{N}$),	
	where  
		\begin{itemize}  
		\item $\forall i\ge0$, $q^{i}\in \mathcal{Q}$, $x^{0}_{q^i}\in\mathcal{X}$ is the initial state at mode $q^i$, $x_0=x^{0}_{q^0}\in \mathcal{X}_0$ is the initial state of the entire trajectory, $x^{i+1}=\xi_{q^{i}}(T^{i};x^{0}_{q^i},u)$ is the initial state at mode $q^{i+1}$;
		\item $\forall i\ge0$, $T^{i}\ge\Lambda_{\rm{min}}>0$ is the dwell time at mode $q^{i}$, while the transition times are $T^0, T^0+T^1, \dots,T^0+T^1+\dots+T^{N_q-1}$;
		\item  $\forall i\ge0$, $(q^{i},q^{i+1})\in \mathcal{E}$.
		\end{itemize} 
\end{definition} 
A trajectory of a switched nominal control system $\mathcal{T}^{\ast}$ can be similarly denoted as a sequence $\rho^{\ast}=\{(q^{i},\xi^{\ast}_{q^{i}}(t;x^{\ast0}_{q^i},u),T^{i})\}_{i=0}^{N_{q}}$ ($N_{q}\in\mathbb{N}$).
	
\begin{definition}[Output Trajectory]
	\label{traj}
For a trajectory $\rho=\{(q^{i},\xi_{q^{i}}(t;x^{0}_{q^i},u),T^{i})\}_{i=0}^{N_{q}}$, we define the output trajectory $s_\rho(\cdot;x_0,u)$ (where we denote $x_0\triangleq x^{0}_{q^0}$ for brevity) as follows:
\[
 	s_\rho(t;x_0,u) =\begin{cases}
 	\xi_{q^0}(t;x_{0},u), ~~~~~~~~~~~~~~~~~~~~~~~~\mbox{if $t<T^0$},\\
 	\xi_{q^i}(t-\displaystyle{\sum_{k=0}^{i-1}}T^k;x^{0}_{q^i},u),\\
 	~~~~~~~~~~\mbox{if $\displaystyle{\sum_{k=0}^{i-1}}T^k\le t<\sum_{k=0}^{i}T^k$, $1\le i\le N_{q}$}.
 	\end{cases} \\
\]                         
\end{definition} 

The output trajectory of a trajectory $\rho^{\ast}=\{(q^{i},\xi^{\ast}_{q^{i}}(t;x^{\ast0}_{q^i},u),T^{i})\}_{i=0}^{N_{q}}$ of a switched nominal control system is denoted as $s_{\rho^{\ast}}(\cdot;x^\ast_0,u)$.
	
\subsection{Metric Temporal Logic (MTL)}
	\label{MTL}   
In this subsection, we briefly review the MTL that are interpreted over continuous-time signals~\cite{FAINEKOScontinous}, the MTL interpreted over discrete-time signals can be found in~\cite{FainekosMTL}. The continuous state of the system we are studying is described by a set of $n$ variables that can be written as a vector $x = \{x_1, x_2, \dots, x_n\}$.
The domain of $x$ is denoted by $\mathcal{X}$. The domain $\mathbb{B} = \{true,false\}$ is the Boolean domain and the time set
is $\mathbb{T} = \mathbb{R}$. The output trajectory $s_\rho(\cdot;x_0,u)$ of a switched system is defined in Sec. \ref{switch}. A set $AP=\{\pi_1,\pi_2,\dots \pi_n\}$ is a set of atomic propositions, each mapping $\mathcal{X}$ to $\mathbb{B}$. The syntax of MTL is defined recursively as follows: 
    \[
	\varphi:=\top\mid \pi\mid\lnot\varphi\mid\varphi_{1}\wedge\varphi_{2}\mid\varphi_{1}\vee
	\varphi_{2}\mid\varphi_{1}\mathcal{U}_{\mathcal{I}}\varphi_{2},   
	\]                           
where $\top$ stands for the Boolean constant True, $\pi$ is an atomic
proposition, $\lnot$ (negation), $\wedge$ (conjunction), $\vee$ (disjunction)
are standard Boolean connectives, $\mathcal{U}$ is a temporal operator
representing \textquotedblleft until\textquotedblright, $\mathcal{S}$ is a past temporal operator
representing \textquotedblleft since\textquotedblright, $\mathcal{I}$ is a time interval of
the form $\mathcal{I}=[i_{1},i_{2}]~(i_{1},i_{2}\in \mathbb{R}_{\geqslant 0}, i_{1}\le i_{2})$. From \textquotedblleft
until\textquotedblright($\mathcal{U}$), we can derive the temporal operators \textquotedblleft
eventually\textquotedblright~$\Diamond_{\mathcal{I}}\varphi=\top\mathcal{U}_{\mathcal{I}}\varphi$ and
	\textquotedblleft always\textquotedblright~$\Box_{\mathcal{I}}\varphi=\lnot\Diamond_{\mathcal{I}}\lnot\varphi$.
	
	We define the set of states that satisfy the atomic proposition $\pi$ as $\mathcal{O}(\pi)\in \mathcal{X}$. For a set $S\subseteq\mathcal{X}$, we define the signed distance from $x$ to $S$ as
	\begin{equation}
	\textbf{Dist$_d(x,S)\triangleq$}
	\begin{cases}
	-\textrm{inf}\{d(x, y)\vert y\in cl(S)\},& \mbox{if $x$ $\not\in S$};\\  
	\textrm{inf}\{d(x, y)\vert y\in\mathcal{X}\setminus S\}, & \mbox{if $x$ $\in S$},
	\end{cases}                        
	\label{sign}
	\end{equation}
where $d$ is a metric on $\mathcal{X}$ and $cl(S)$ denotes the closure of the set $S$. In this paper, we use the metric $d(x,y)=\norm{x-y}$, where $\left\Vert\cdot\right\Vert $ denotes the 2-norm. 

We use $\left[\left[\varphi\right]\right](s_\rho(\cdot;x_0,u), t)$ to denote the robustness degree of the output trajectory $s_\rho(\cdot;x_0,u)$ with respect to the formula $\varphi$ at time $t$. We denote $-\mathcal{I}\triangleq[-i_{2},-i_{1}]$ when $\mathcal{I}=[i_{1},i_{2}]$. The robust semantics of a formula $\varphi$ with respect to $s_\rho(\cdot;x_0,u)$ are defined recursively as follows~\cite{Dokhanchi2014}:
	\begin{align}  
	\begin{split}
	\left[\left[\top\right]\right](s_\rho(\cdot;x_0,u), t):=& +\infty,\\
	\left[\left[\pi\right]\right](s_\rho(\cdot;x_0,u), t):=&\textbf{Dist$_d(s_\rho(\cdot;x_0,u)(t),\mathcal{O}(\pi))$},\\
	\left[\left[\neg\varphi\right]\right](s_\rho(\cdot;x_0,u), t):=&-\left[\left[ \varphi\right]\right](s_\rho(\cdot;x_0,u), t),\\
	\left[\left[\varphi_1\wedge\varphi_2\right]\right](s_\rho(\cdot;x_0,u), t):=&\min\big(\left[\left[ \varphi_1\right]\right](s_\rho(\cdot;x_0,u), t),\\&\left[\left[\varphi_2\right]\right](s_\rho(\cdot;x_0,u), t)\big),\\
	\left[\left[\varphi_1\mathcal{U}_{\mathcal{I}}\varphi_{2}\right]\right](s_\rho(\cdot;x_0,u), t):=&\max_{t'\in (t+\mathcal{I})}\Big(\min\big(\left[\left[\varphi_2\right]\right](s_\rho(\cdot;x_0,u),\\&  t'), \min_{t\le t''<t'}\left[\left[\varphi_1\right]\right]
	(s_\rho(\cdot;x_0,u),t'')\big)\Big).	 
	\end{split}
	\end{align}

 If the output trajectory $s_\rho(\cdot;x_0,u)$ is only defined on a finite time interval ($\triangleq \mathbb{T}_o$), then the time domain of the formula $\varphi$ is also finite. It is defined recursively as follows \cite{zhe2016}:                                                                                          
\begin{align}                                                            
\begin{split}
\textcolor{black}{D(\pi, s_\rho(\cdot;x_0,u))}& =\mathbb{T}_o,\\
\textcolor{black}{D(\lnot\varphi, s_\rho(\cdot;x_0,u))} & \textcolor{black}{=D(\varphi, s_\rho(\cdot;x_0,u))},\\
\textcolor{black}{D(\varphi_{1}\wedge\varphi_{2}, s_\rho(\cdot;x_0,u))} & \textcolor{black}{=D(\varphi_{1}, s_\rho(\cdot;x_0,u))}\\&~~~~\textcolor{black}{\cap D(\varphi_{2}, s_\rho(\cdot;x_0,u))},\\
\textcolor{black}{D(\varphi_{1}\mathcal{U}_{\mathcal{I}}\varphi_{2}, s_\rho(\cdot;x_0,u))} & =\{t\vert t+\mathcal{I} \subset (D(\varphi_{1}, s_\rho(\cdot;x_0,u))\\&~~~~\cap D(\varphi_{2}, s_\rho(\cdot;x_0,u)))\}.
\end{split}             
\end{align}
                        
\section{Stochastic Control Bisimulation Function}
\subsection{Stochastic Control Bisimulation Function}
We consider the switched stochastic control system with the following dynamics in the mode $q$:
\begin{align}
\begin{split}
& dx=f_q(x,u)dt+g_q(x,u)dw, 
\end{split}
\label{sys}
\end{align}
where the state $x\in\mathcal{X}\in\mathbb{R}^{n}$, the input $u\in\mathcal{V}\in\mathbb{R}^{p}$, $w$ is an $\mathbb{R}^{m}$-valued standard Brownian motion.

Note that the dynamics is essentially the same as that in \cite{Julius2008CDC} when the input signal $u(\cdot)$ is given and bounded, where the existence and uniqueness of the solution of (\ref{sys}) can be guaranteed with the conditions given in \cite{Julius2008CDC}. 

We also consider the switched nominal control system in the mode $q$ as the nominal deterministic version:
\begin{align}
& dx^{\ast}=f_q(x^{\ast},u)dt,  
\label{nom}
\end{align}
 
	\begin{definition}	
		A continuously differentiable function $\psi_q:\mathcal{X}\times\mathcal{X}\times\mathbb{T}\rightarrow \mathbb{R}_{\geqslant 0}$ is a \textbf{time-varying control autobisimulation function} of
		the switched nominal control system (\ref{nom}) in the mode $q$ if for any $x, \tilde{x} \in \mathcal{X}$ ($x\neq\tilde{x}$) and any $t\in\mathbb{T}$ there exists a function $u:\mathbb{R}^n\times\mathbb{T}\rightarrow\mathbb{R}^p$ such that $\psi_q(x,\tilde{x},t)>0$, $\psi_q(x,x,t)=0$ and $\frac{\partial{\psi_q(x,\tilde{x},t)}}{\partial{x}}f_q(x,u(x,t))+\frac{\partial{\psi_q(x,\tilde{x},t)}}{\partial{\tilde{x}}}f_q(\tilde{x},u(\tilde{x},t))+\frac{\partial{\psi_q(x,\tilde{x},t)}}{\partial{t}}\leq0$.
	\end{definition} 
	
In the following, we extend the concept of control autobisimulation function to the stochastic setting.
	
	\begin{definition}	
		A twice differentiable function $\phi_q$: $\mathcal{X}\times\mathcal{X}\rightarrow \mathbb{R}_{\geqslant 0}$ is a \textbf{stochastic control bisimulation function} between (\ref{sys}) and its nominal system (\ref{nom}) if it satisfies \cite{zheACC2018wind}
		\begin{align}
		\begin{split}
		&\phi_q(x,\tilde{x},t)>0, \forall x,\tilde{x}\in \mathcal{X},  x\neq\tilde{x}, \forall t\in \mathbb{T},\\
		&\phi_q(x, x, t)=0, \forall x\in \mathcal{X}, \forall t\in \mathbb{T}, 
		\end{split}
		\end{align}
		and there exist $\alpha_q>0$ and a function
		$u:\mathbb{R}^n\times\mathbb{T}\rightarrow \mathbb{R}^p$ such that 
		\begin{align}
		\begin{split}
		&\frac{\partial\phi_q(x,\tilde{x},t)}{\partial x}f_q(x,u(x,t))+\frac{\partial\phi_q(x,\tilde{x},t)}{\partial \tilde{x}}f_q(\tilde{x},u(\tilde{x},t))\\&+\frac{\partial\phi_q(x,\tilde{x},t)}{\partial t}+\frac{1}{2}g_q^{T}(x,u(x,t))\frac{\partial^{2}\phi_q(x,\tilde{x},t)}{\partial x^{2}}g_q(x,u(x,t))\\&<\alpha_q, 
		\end{split}
		\end{align}
		for any $x, \tilde{x}\in\mathcal{X}$. 
	\end{definition}	

The stochastic control bisimulation function establishes a bound between the trajectories of system (\ref{sys}) and its nominal system (\ref{nom}). 

\subsection{Stochastic Control Bisimulation Function for Switched Linear Dynamics}	
In this subsection, we consider the switched stochastic control system with the following linear dynamics in the mode $q$:
\begin{align}
\begin{split}
& dx=(A_qx+B_qu)dt+\Sigma_q dw,                                                                                      
\end{split}
\label{syslinear}
\end{align}
where $A_q\in\mathbb{R}^{n\times n}$, $B_q\in\mathbb{R}^{n\times p}$, $\Sigma_q\in\mathbb{R}^{n\times m}$.

We can construct a stochastic control bisimulation function of the form
	\begin{center}
		$\phi_q(x, \tilde{x},t)=(x-\tilde{x})^{T}M_q(x-\tilde{x})e^{\mu_q t}$, 
	\end{center}
where $M_q$ is a symmetric positive definite matrix. In order for this function to qualify as a stochastic control bisimulation function, we need to have $M_q\succ 0$, and
	\begin{align}
	\begin{split}
	&\frac{\partial\phi_q(x, \tilde{x}, t)}{\partial x}(A_qx+B_qu)+\frac{\partial\phi_q(x,\tilde{x},t)}{\partial \tilde{x}}(A_q\tilde{x}+B_qu)\\ &+\frac{\partial\phi_q(x,\tilde{x},t)}{\partial t}+
	trace(\frac{1}{2}\Sigma_q^{T}(\frac{\partial^{2}\phi_q(x,\tilde{x},t)}{\partial x^{2}})\Sigma_q)\\ =&(x-\tilde{x})^{T}(2M_qA_q+\mu M_q)(x-\tilde{x})+trace(\Sigma_q^{T}M_q\Sigma_q)\\
	<&\alpha.
	\label{func}
	\end{split}
	\end{align}
	for some $\alpha_q>0$. If we pick $\alpha_q=trace(\Sigma_q^{T}M_q\Sigma_q)$, the inequality (\ref{func}) becomes a linear matrix inequality (LMI)
	\begin{align}
		& A_q^{T}M_q+M_qA_q+\mu M_q\prec 0. 
		\label{LMI} 
	\end{align}
We denote the system trajectory starting from $x_0$ with the input signal $u(\cdot)$ as $\xi({\bm\cdot};x_0,u)$. It can be seen that (\ref{func}) holds for any input signal $u(\cdot)$, so $u(\cdot)$ is free to be designed.
It can also be verified that $\psi_q(x,\tilde{x},t)=\phi_q(x, \tilde{x},t)=(x-\tilde{x})^{T}M_q(x-\tilde{x})$ is also a time-varying control autobisimulation function of the nominal system  
\begin{align}
\begin{split}
& dx^{\ast}=(A_qx^{\ast}+B_qu)dt. 
\end{split}
\label{nomlinear}
\end{align}

\begin{proposition}                                     
	Given the dynamics of (\ref{nomlinear}), $\psi_q(x, \tilde{x},t) = (x-\tilde{x})^TM_q(x-\tilde{x})e^{\mu t}$ is an autobisimulation function if the matrix $M_q$ satisfies the following:
	\begin{align}
	\begin{split}
	&~~~~~~~~M_q\succ 0, ~A_q^TM_q+M_qA_q+\mu M_q\preceq0. 
	\end{split}                                   
	\end{align}                  
\end{proposition}

\begin{proof} 
	As $e^{\mu t}>0$, if $M_q\succ 0$, then for any $x, \tilde{x}\in\mathcal{X}$ and any $t$, we have $\psi_q(x, \tilde{x},t) = (x-\tilde{x})^TM_q(x-\tilde{x})e^{\mu t}>0$. If $A_q^TM_q+M_qA_q+\mu M_q\preceq0$, we have for any $x, \tilde{x}\in\mathcal{X}$ and any $t$, 
	\begin{align}\nonumber
	\begin{split}
	&\frac{\partial{\psi_q(x, \tilde{x},t)}}{\partial{x}}f_q(x)+\frac{\partial{\psi_q(x, \tilde{x},t)}}{\partial{\tilde{x}}}f_q(\tilde{x})+\frac{\partial{\psi_q(x, \tilde{x},t)}}{\partial{t}}\\
	&=(x-\tilde{x})^T(A_q^TM_q+M_qA_q+\mu M_q)(x-\tilde{x})e^{\mu t}\le 0.
	\end{split}                                   
	\end{align}   
	So $\psi_q(x, \tilde{x},t) = (x-\tilde{x})^TM_q(x-\tilde{x})e^{\mu t}$ is an autobisimulation function of system (\ref{nomlinear}).    
\end{proof}   

We denote the output trajectory of the nominal system starting from $x_0$ with the input signal $u(\cdot)$ as $s_{\rho^{\ast}}({\bm\cdot};x_0,u)$. 

\begin{proposition}
If $\phi_q$ is a stochastic control bisimulation function between the switched stochastic control system (\ref{syslinear}) and its switched nominal control system (\ref{nomlinear}) in the mode $q$, then for any $T>0$,
\begin{align}
& P\{\sup_{0\leq t\leq T}\phi_q(\xi_q^{\ast}(t;x^0_q,u), \xi_q(t;x^0_q,u))<\gamma\}> 1-\frac{\alpha_q T}{\gamma}.
\label{prob}
\end{align}
\end{proposition}

  \begin{proof} 
  Straightforward from Proposition 2.2 of \cite{Julius2008CDC} and (\ref{func}).
  \end{proof}  	
  
It can be seen from (\ref{prob}) that $\phi_q$ provides a probabilistic upper bound for the distance between the states of the switched stochastic control system and its switched nominal control system in the node $q$ in a finite time horizon. We denote $B_{q}(x,\gamma)\triangleq\{\tilde{x}\vert(x-\tilde{x})^TM_q(x-\tilde{x})\le\gamma\}$.
	
\section{Stochastic Controller Synthesis}
	\subsection{Feedforward Controller Synthesis}
	\label{Sec_Feedforward}
We denote the set of states that satisfy the predicate $p$ as 
$\mathcal{O}(p)\subset\mathcal{X}$. In this paper, we consider a fragment of MTL formulas in the following form:
\begin{align} 
\begin{split}
\varphi =& \Box_{[\tau_1,T_{\textrm{end}}]}p_1 \wedge \Box_{[\tau_2,T_{\textrm{end}}]}p_2 \wedge\dots \wedge\Box_{[\tau_{\eta},T_{\textrm{end}}]}p_{\eta},
\end{split}
\label{MTLform1}
\end{align}	
where $\tau_1<\tau_2<\dots\tau_{\eta}\leq T_{\textrm{end}}$, $T_{\textrm{end}}$ is the end of the simulation time, $\mathcal{O}(p_{\eta})\subset\mathcal{O}(p_{\eta-1})\subset\dots\subset\mathcal{O}(p_1)$, each predicate $p_k$ is in the following form:     
\begin{equation}
p_k\triangleq\left(\bigwedge_{\nu=1}^{n_k}a_{k,\nu}^{T}x+c_{k,\nu}^{T}u<b_{k,\nu}\right),
\label{predicate}
\end{equation}
where ${a}_{k,\nu}\in\mathbb{R}^{n}$ and $b_{k,\nu}\in\mathbb{R}$ denote the                    
parameters that define the predicate, $n_k$ is the number of atomic predicates in the           
$k$-th predicate. We constraint $\lVert{a}_{k,\nu}\rVert_{2}=1$ to reduce redundancy.

The MTL formulas in the above-defined form is actually specifying a series of regions to be entered before certain deadlines and stayed thereafter, with larger regions corresponding to tighter deadlines, as shown in Fig. \ref{deadline}. The MTL formulas in this form is especially useful in power system frequency regulations as discussed in Section \ref{sec_power}. 
	
	\begin{figure}
		\centering
		\includegraphics[scale=0.25]{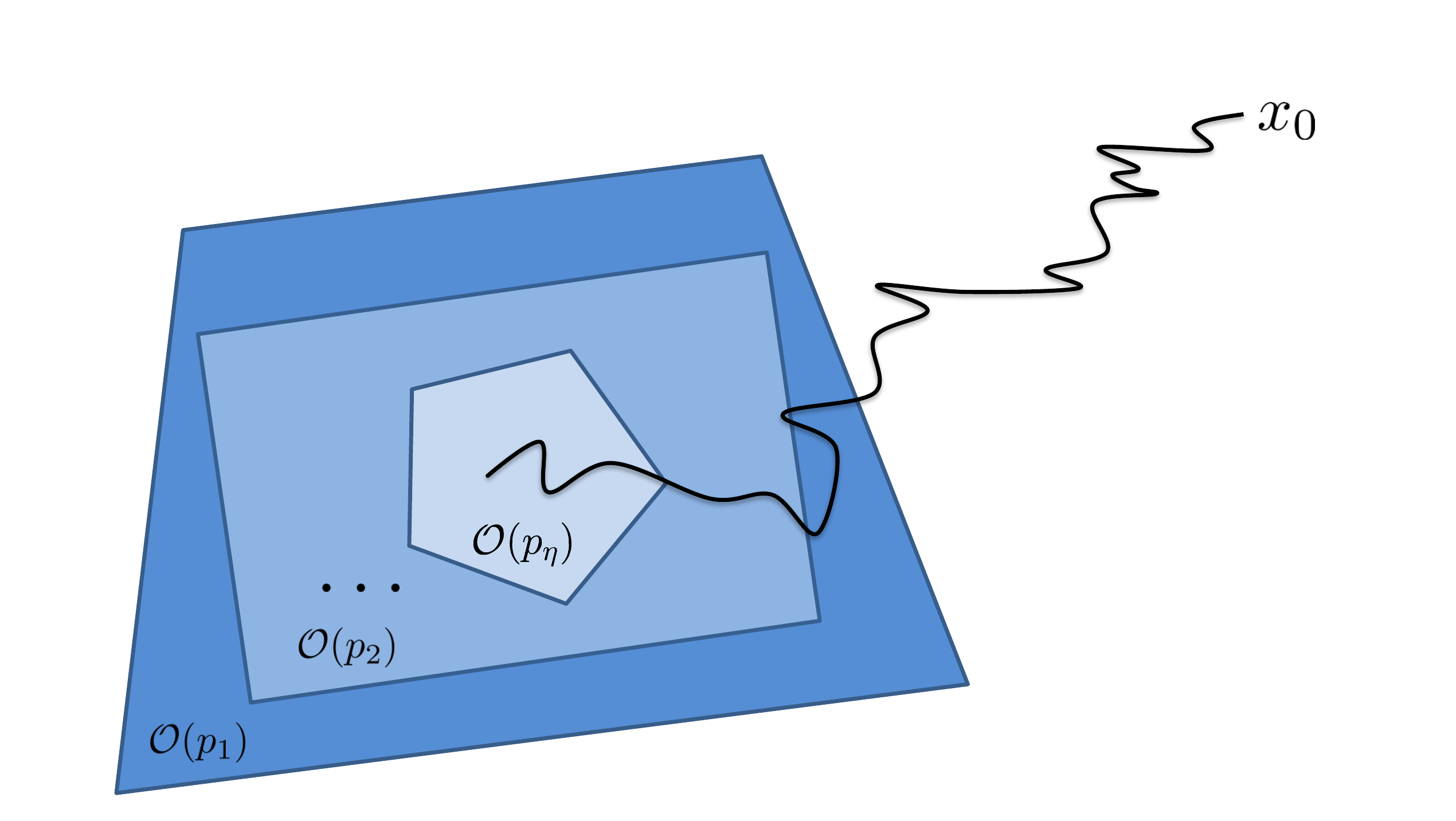}
		\caption{Illustration of the MTL formula $\varphi$ in (\ref{MTLform1}).} 
		\label{deadline}
	\end{figure}
	
The $\delta_{k,\nu}$-robust modified formula $\hat{\varphi}_{\delta}$ is defined as follows:
	\begin{align}
\hat{\varphi}_{\delta}\triangleq& \Box_{[\tau_1,T_{\textrm{end}}]}\hat{p}_1 \wedge \Box_{[\tau_2,T_{\textrm{end}}]}\hat{p}_2 \wedge\dots \wedge\Box_{[\tau_{\eta},T_{\textrm{end}}]}\hat{p}_{\eta},
\label{MTLform2}
	\end{align}
where each predicate $\hat{p}_k$ is modified from (\ref{predicate}) as follows:
\begin{equation}
\hat{p}_k\triangleq\left(\bigwedge_{\nu=1}^{n_k}a_{k,\nu}^{T}x+c_{k,\nu}^{T}u<b_{k,\nu}-\Delta_{k,\nu}(t)\right),
\label{predicate2}
\end{equation} 
where
\[
 	\Delta_{k,\nu}(t) =\begin{cases}
 	\delta^0_{k,\nu}e^{-\mu t/2}, ~~~~~~~~~~~~~~~~~~~~~~~~\mbox{if $t<T^0$},\\
 	\delta^i_{k,\nu}e^{-\mu (t-\sum_{j=0}^{i-1}T^j)/2},\\
 	~~~~~~~~~~\mbox{if $\sum_{j=0}^{i-1}T^j\le t<\sum_{j=0}^{i}T^j$, $1\le i\le N_{q}$}.
 	\end{cases} \\
\]                                                                     
	
	\begin{theorem}	                                                        
		If for every $k\in\{1,\dots,\eta\}$ and $\nu\in\{1,\dots,n_k\}$, there exist $z^i_{k,\nu} (i=0,1,\dots,N_q), \epsilon>0$ such that $(z_{k,\nu}^{i})^2a_{k,\nu}a_{k,\nu}^{T}\preceq M_{q^i}$ and $\left[\left[\varphi_{\hat{\delta}}\right]\right](s_{\rho^{\ast}}({\bm\cdot};x^{\ast}_0,u), 0)\ge 0$, where  $\rho^{\ast}=\{(q^{i},\xi^{\ast}_{q^{i}}(t;x^{\ast 0}_{q^i},u),T^{i})\}_{i=0}^{N_{q}}$ is a trajectory of the nominal system, $\varphi_{\hat{\delta}}$ is the $\hat{\delta}_{k,\nu}$-robust modified formula of $\varphi$, $\hat{\delta}^i_{k,\nu}=(\sqrt{r_{q^i}}+\sqrt{\hat{\gamma}})/z^i_{k,\nu}$, $\hat{\gamma}=\frac{(\max\limits_{i}\alpha_{q^i})\cdot T_{\textrm{end}}}{\epsilon}$, $B_{q^{i-1}}(\xi_{q^{i-1}}(\tau;x^{\ast 0}_{q^{i-1}},u),r_{q^{i-1}}e^{-\mu_{q^{i-1}}T^{i-1}/2})\subset B_{q^{i}}(x^{\ast 0}_{q^i},r_{q^{i}})$, then for any $\tilde{x}_0\in B_{q^0}(x^{\ast}_0,r_{q^0})$, the output trajectory $s_{\tilde{\rho}}({\bm\cdot};\tilde{x}_0,u)$ of trajectory $\tilde{\rho}=\{(q^{i},\xi_{q^{i}}(t;\tilde{x}^{0}_{q^i},u),T^{i})\}_{i=0}^{N_{q}}$ satisfies MTL specification $\varphi$ with probability at least $1-\epsilon$, i.e. $P\{\left[\left[\varphi\right]\right](s_{\tilde{\rho}}({\bm\cdot};\tilde{x}_0,u), 0)\ge 0\}>1-\epsilon$.
		\label{th1}
	\end{theorem}	
	\begin{proof} 
		See Appendix.
	\end{proof}  	
	                       
From Theorem \ref{th1}, if we can design the input signal $u(\cdot)$ such that the nominal trajectory $s_{\rho^{\ast}}({\bm\cdot};x_0,u)$ of the nominal system (\ref{nom}) satisfies the $\hat{\delta}_{k,\nu}$-robust modified formula of $\varphi$ ($\hat{\delta}^i_{k,\nu}\triangleq(\sqrt{r_{q^i}}+\sqrt{\hat{\gamma}})/z^i_{k,\nu}$), then all the trajectories of the stochastic system (\ref{sys}) starting from the initial set $B_{q^0}(x_0,r_{q^0})$ are guaranteed to satisfy the MTL specification $\varphi$ with probability at least $1-\epsilon$. To make the robust modification as tight as possible, for every $k\in\{1,\dots,\eta\}$ and $\nu\in\{1,\dots,n_k\}$, we compute the maximal $z^i_{k,\nu}$ such that $(z_{k,\nu}^{i})^2a_{k,\nu}a_{k,\nu}^{T} \preceq M_{q^{i}}$. We denote the maximal value of $z^i_{k,\nu}$ as $z^{i\ast}_{k,\nu}$, $\hat{\delta}^{i\ast}_{k,\nu}\triangleq(\sqrt{r_{q^{i}}}+\sqrt{\hat{\gamma}})/z^{i\ast}_{k,\nu}$, and the $\hat{\delta}^{\ast}_{k,\nu}$-robust modified formula as $\varphi_{\hat{\delta}^{\ast}}$ (the predicates in $\varphi_{\hat{\delta}^{\ast}}$ are denoted as $\hat{p}^{\ast}_k$). 
	
The optimization problem to find the optimal input signal such that the nominal trajectory satisfies the $\delta^{\ast}_{k,\nu}$-robust modified formula $\varphi_{\hat{\delta}^{\ast}}$ is formulated as follows: 
	\begin{align}
	\begin{split}
	\underset{u(\cdot)}{\argmin} ~ & J(u(\cdot)) \\
	\text{subject to} ~ &\left[\left[\varphi_{\hat{\delta}^{\ast}}\right]\right](s_{\rho^{\ast}}({\bm\cdot};x^{\ast}_0,u), 0)\ge 0.
	\end{split}
	\label{opt}
	\end{align}

The performance measure $J(u(\cdot))$ can be set as the control effort $\norm{u(\cdot)}_2$ or $\norm{u(\cdot)}_1$. For linear systems, the above optimization problem can be converted to a linear program (LP) problem~\cite{BluSTL,sayan2016} and it can be solved efficiently by LP solvers. As we only focus on the time horizon $[0, T_{\textrm{end}}]$ (e.g. the most critical time period in the transient response of the power system), we assume that some other controllers that make $\mathcal{O}(\hat{p}^{\ast}_{\eta})$ a control invariant set for the nominal trajectories will take over after $T_{\textrm{end}}$. 

\subsection{Feedback Controller Synthesis}	
\label{feedbackcontrol}
In this section, we design a feedback control law to replace the optimal input signals of the feedforward controller. The advantage of a feedback controller is that it is more robust to unexpected disturbances. We simulate $N$ trajectories $\rho^{\ast}_{\ell}=\{(q^{i},\xi^{\ast}_{q^{i},\ell}(t;x^{\ast0}_{q^i,\ell},u_{q^i,\ell}),T^{i})\}_{i=0}^{N_{q}}$ of the switched nominal control system ($\ell=1,2,\dots,N$). When the states and inputs of the trajectories are calculated using numeric simulators such as ODE or CVODE, the data are discrete and therefore in the following we use $\xi^{\ast}_{q,\ell}[j^q_t]\triangleq\xi^{\ast}_{q,\ell}(j^q_t;x^{0\ast}_{q,\ell},u_{q,\ell})$ and $u_{q,\ell}[j^q_t]$ ($\ell=1,2,\dots,N$, $j^q_t=0,1,\dots,N^q_t$) to denote the flow solution and the input of the $\ell$th nominal trajectory of the nominal system (\ref{nom}) at the $j^q_t$th time instant of mode $q$, respectively. In order to account for the situation when unexpected disturbances occur, we resolve (\ref{opt}) with a different switching time instant $T^q+\varrho$ in each mode $q$, where $T^q$ is the duration time in mode $q$ for the nominal trajectories, $\varrho\in(0,T^q)$ is a positive number. The obtained control input in the mode $q$ is denoted as $u'_{q,\ell}$, and the corresponding extended time instants are denoted as $N^q_t+1,\dots,N^{'q}_t$.

The algorithm to generate the feedback law is shown in Algorithm \ref{alg}. We apply the following feedback law $\chi_u(x,t)$ which depends both on the current state $x$, the current mode $q$ and the time $t[j^q_t]$ spent in mode $q$:\\
$\chi_u(x,q,t[j^q_t])\triangleq$\\
\begin{equation}
\begin{cases}
u_{q,1}[0],& \mbox{if $(x,t[j^q_t])\in \mathcal{X}^q_{0,1}[j^q_t]\times[0,t[0]]$},\\   
~~~~~~\vdots  & ~~~~~~\vdots\\
u_{q,N}[0],& \mbox{if $(x,t[j^q_t])\in \mathcal{X}^q_{0,N}[j^q_t]\times[0,t[0]]$},\\
~~~~~~\vdots  & ~~~~~~\vdots\\		
u_{q,1}[i^q_t],& \mbox{if $(x,t[j^q_t])\in \mathcal{X}^q_{i^q_t,1}[j^q_t]\times[t[i^q_t]-\varrho,t[i^q_t]]$},\\   
~~~~~~\vdots  & ~~~~~~\vdots\\
u_{q,N}[i^q_t],& \mbox{if $(x,t[j^q_t])\in \mathcal{X}^q_{i^q_t,N}[j^q_t]\times[t[i^q_t]-\varrho,t[i^q_t]]$},\\
~~~~~~\vdots  & ~~~~~~\vdots\\
u_{q,1}[N^q_t],& \mbox{if $(x,t[j^q_t])\in \mathcal{X}^q_{N^q_t,1}[j^q_t]\times[T^q-\varrho,T^q]$},\\   
~~~~~~\vdots  & ~~~~~~\vdots\\
u_{q,N}[N^q_t],& \mbox{if $(x,t[j^q_t])\in \mathcal{X}^q_{N^q_t,N}[j^q_t]\times[T^q-\varrho,T^q]$},\\
~~~~~~\vdots  & ~~~~~~\vdots\\
u'_{q,1}[N^{'q}_t],& \mbox{if $(x,t[j^q_t])\in \mathcal{X}^q_{N^{'q}_t,1}[j^q_t]\times[T^q,T^q+\varrho]$},\\   
~~~~~~\vdots  & ~~~~~~\vdots\\
u'_{q,N}[N^{'q}_t],& \mbox{if $(x,t[j^q_t])\in \mathcal{X}^q_{N^{'q}_t,N}[j^q_t]\times[T^q,T^q+\varrho]$},\\
u_{\hat{\ell}}[\hat{j^q_t}],& \mbox{otherwise},					
\end{cases}                       
\label{cases} 
\end{equation}
where each region $\mathcal{X}^q_{i^q_t,\ell}[j^q_t]$ is changing with time and defined as follows:  
$\mathcal{X}^q_{i^q_t,\ell}[j^q_t]=$                                                                                
\begin{equation}                                                            
\begin{cases}
\hat{\mathcal{X}}^{q}_{i^q_t,\ell}[j^q_t]\setminus\bigcup\limits_{i'_t>i^q_t}\hat{\mathcal{X}}^{q}_{i'_t,\ell}[j^q_t],~~~~~~~~~~~~~~~~~~~~~~~~\mbox{if $\ell=\hat{\ell}[j^q_t]$}, \\   
B_{q}(\xi^{\ast}_{\ell}[i^q_t],\hat{r}_q[j^q_t])\setminus\Big(\big(\bigcup\limits_{i^{q'}_t}\hat{\mathcal{X}}^{q}_{i^{q'}_t,\hat{\ell}}\big)\cup\big(\bigcup\limits_{\ell'>\ell}\bigcup\limits_{i^{q'}_t}B_{q}(\xi^{\ast}_{\ell'}[i^{q'}_t],\hat{r}_q[j^q_t])\big)\\ \cup\big(\bigcup\limits_{i^{q'}_t>i^q_t}B_{q}(\xi^{\ast}_{\ell}[i^{q'}_t],\hat{r}_q[j^q_t])\big)\Big), ~~~~~~~~~~~~~~~\mbox{if $\ell\neq\hat{\ell}[j^q_t]$}.					
\end{cases}        
\label{init}
\end{equation}
where $\hat{r}_q[j^q_t]\triangleq r_qe^{-\mu_{q}t[j^q_t]/2}$, $\hat{\ell}$ is initially assigned as $\min\{\ell\vert \tilde{x}_0\in B_{q^{0}}(x^{\ast}_{0,{\ell}},r)\}$ and then assigned according to lines 5-9 of Algorithm \ref{alg}, $\hat{j^q_t}$ is initially assigned as 0 and then assigned according to lines 5-9 of Algorithm \ref{alg}, $\hat{\mathcal{X}}_{i^q_t,\ell}[j^q_t]=\{x\vert\phi(x,\tilde{x})\le\hat{\gamma},\tilde{x}\in B_{q}(\xi^{\ast}_{\ell}[i^q_t],\hat{r}_q[j^q_t])\}$ is defined as the \textbf{stochastic robust neighbourhood with probability ($1-\epsilon$)}, where $\hat{\gamma}=\frac{\alpha T_{\textrm{end}}}{\epsilon}$.  

\begin{algorithm}
	\begin{algorithmic}[1]
		\caption{Feedback law generation.}				
		\State $\hat{\ell}[1]\gets\min\{\ell\vert \tilde{x}_0\in B_{q^{0}}(x^{\ast}_{0,{\ell}},r)\}$	
		\State For the current mode $q$ and every $i^q_t$, $\ell$, obtain $\mathcal{X}_{i^q_t,\ell}[1]$ according to (\ref{init})  
		\For{each mode $q$ in the order that the nominal trajectories pass} 
		\State $j^q_t\gets 0$, $\hat{j^q_t}\gets 0$							    
		\While{$j^q_t<N^q_t$ ($q$ is the current mode)}	
		\If{there exists $\ell'\neq\hat{\ell}[j^q_t]$ and $i^q_t\ge j^q_t$ such that $x\in \mathcal{X}_{i^q_t,\ell'}[j^q_t]$}
		\State $\hat{\ell}[j^q_t+1]\gets\ell'$, $\hat{j^q_t}\gets i^q_t+1$		
		\Else
		\State $\hat{\ell}[j^q_t+1]\gets\hat{\ell}[j^q_t]$, $\hat{j^q_t}\gets \hat{j^q_t}+1$	
		\EndIf
		\State Obtain $\chi_u(x,q,t[j^q_t])$ according to (\ref{cases})
		\State $j^q_t\gets j^q_t+1$	
	    \State For every $i^q_t$, $\ell$, obtain $\mathcal{X}_{i^q_t,\ell}[j^q_t]$ according to (\ref{init})		
		\EndWhile	  
		\EndFor  
		\label{alg}  
	\end{algorithmic}
\end{algorithm}	

As shown in Fig. \ref{redblue}, for the two possible trajectories (realizations) of the stochastic system, the trajectory (realization) in Fig. \ref{redblue} (a) stays within the stochastic robust neighbourhoods around the nominal trajectory $\xi^{\ast}_{\ell}[\cdot]$ throughout the time while the trajectory (realization) in Fig. \ref{redblue} (b) exits the stochastic robust neighbourhoods around $\xi^{\ast}_{\ell}[\cdot]$ and enters $\mathcal{X}_{i^q_t,\ell'}[j^q_t]$ at time instant $j^q_t$, thus $\hat{\ell}$, $\hat{j^q_t}$ and the following inputs are changed according to Algorithm \ref{alg} such that the trajectory stays within the stochastic robust neighbourhoods around the nominal trajectory $\xi^{\ast}_{\ell'}[\cdot]$ with probability ($1-\epsilon$) since time instant $j^q_t$. We use $s_{\rho}(\cdot;\tilde{x}_0,\chi_u)$ to denote the output trajectory of the stochastic system starting from $\tilde{x}_0$ when the feedback control law $\chi_u(x,t)$ is applied. 

	\begin{theorem}	
		If for every $\ell$, $\left[\left[\varphi_{\hat{\delta}^{\ast}}\right]\right](s_{\rho_{\ell}^{\ast}}({\bm\cdot};x^{\ast}_{0,\ell},u_{\ell}), 0)\ge 0$ ($\varphi_{\hat{\delta}^{\ast}}$ is the $\hat{\delta}^{\ast}_{k,\nu}$-robust modified formula, where $\hat{\delta}^{i\ast}_{k,\nu}=(\sqrt{r_{q}^i}+\sqrt{\hat{\gamma}})/z^{i\ast}_{k,\nu}$, $\hat{\gamma}=\frac{(\max\limits_{i}\alpha_{q^i})\cdot T_{\textrm{end}}}{\epsilon}$), then for any $\tilde{x}_0\in\bigcup\limits_{\ell=1}^{N}B_{q^{0}}(x^{\ast}_{0,{\ell}},r)$ the trajectory $\xi({\bm\cdot};\tilde{x}_0,\chi_u)$ satisfies the MTL specification $\varphi$ with probability at least $1-\epsilon$, i.e. $P\displaystyle\{\left[\left[\varphi\right]\right](s_{\rho}({\bm\cdot};\tilde{x}_0,\chi_u), 0)\ge 0\}>1-\epsilon$.            
		\label{th2}                                      
	\end{theorem}
   \begin{proof} 
   	See Appendix.
   \end{proof}  
   	
			\begin{figure}
				\centering
				\includegraphics[scale=0.3]{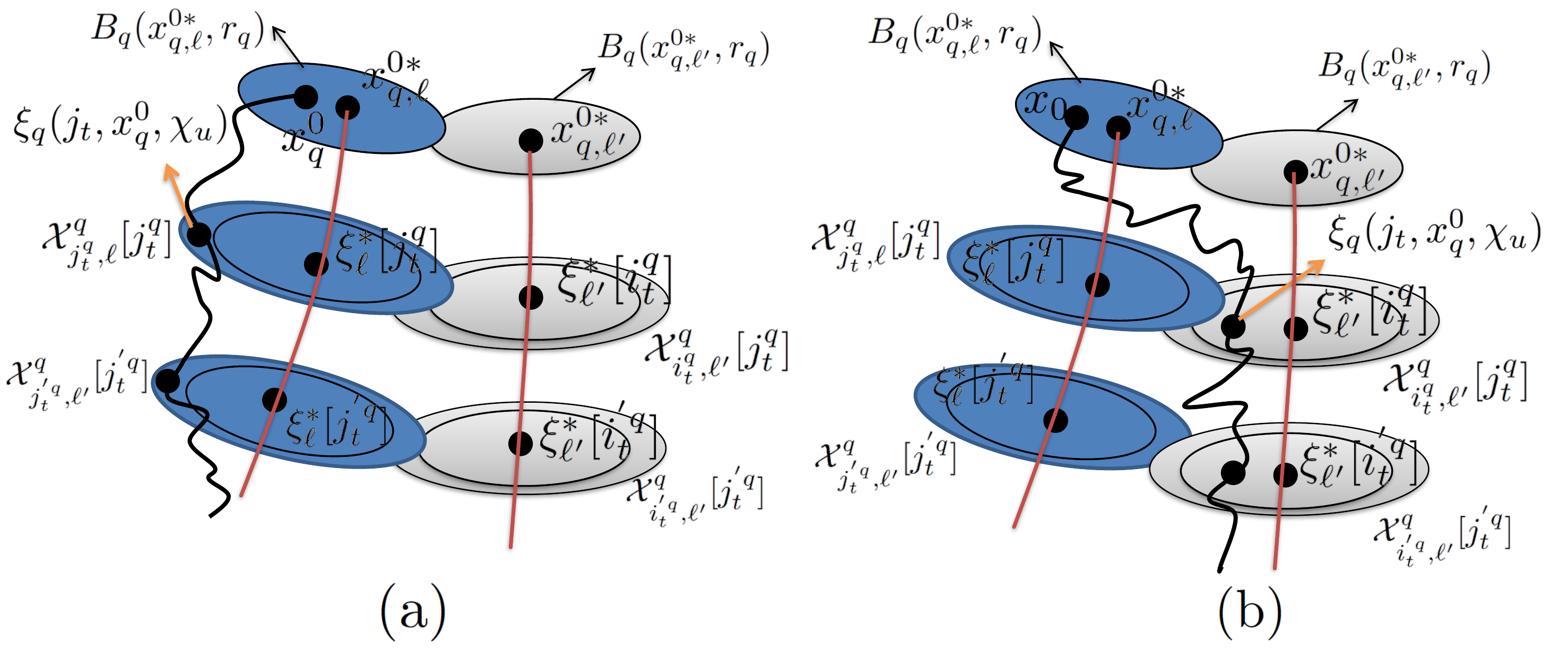}
				\caption{Two possible trajectories (realizations) of the stochastic system (black), two nominal trajectories (brown) and the stochastic robust neighbourhoods around the nominal trajectories.}
				\label{redblue}
			\end{figure}   
While Theorem \ref{th2} gives provable probabilistic guarantees when there is no unexpected disturbance, the following theorem considers the situation when unexpected disturbance occurs.

	\begin{theorem}	
    	Assume that for every $\ell$, $\left[\left[\varphi_{\hat{\delta}^{\ast}}\right]\right]$ $(s_{\rho_{\ell}^{\ast}}({\bm\cdot};x^{\ast}_{0,\ell},u_{\ell}), 0)\ge 0$ ($\varphi_{\hat{\delta}^{\ast}}$ is defined in the same way as in Theorem \ref{th2}). For any $\tilde{x}_0\in\bigcup\limits_{\ell=1}^{N}B_{q^{0}}(x^{\ast}_{0,{\ell}},r)$, if at a time instant $j^q_t$ in mode $q$, unexpected disturbances can perturb the state $x$ to another state $x'\in\bigcup\limits_{\ell=1}^{N}B_{q^{\hat{i}}}(\xi^{\ast}_{\ell}[i^q_t],r_{q^{\hat{i}}})$, where $j^q_t\le i^q_t\le j^q_t+\varrho$, and $\forall i, \bigcup_{t[j^{q^{i-1}}_t]\in T^{i-1}+[0,\varrho]}B_{q^{i-1}}(\xi^{\ast}_{q^{i-1}}[j^{q^{i-1}}_t],\hat{r}_{q^{i-1}}[j^{q^{i-1}}_t])\subset B_{q^{i}}(x^{\ast 0}_{q^i},r_{q^{i}})$, then the trajectory $s_{\rho}({\bm\cdot};\tilde{x}_0,\chi_u)$ still satisfies the MTL specification $\varphi$ with probability at least $1-\epsilon$.                  
    	\label{th3}                                
	\end{theorem}
	
  \begin{proof} 
  	See Appendix.
  \end{proof}

\section{Wind Turbine Generator Controller Synthesis}
\label{sec_power}
\subsection{Case Study I}
	In this section, we apply the controller synthesis method in designing a coordinated controller for regulating the grid frequency of a four-bus system with a 600 MW thermal plant $\textrm{G}_{1}$
	made up of four identical units, a wind farm $\textrm{G}_{2}$ consisting of 200 identical 1.5 MW Type-C wind turbine generators (WTG) and an energy storage system (ESS), as shown in Fig. \ref{dfig}. The configuration parameters of each Type-C WTG can be found in Appendix B of \cite{PulgarPhd} (we set $C_{\mathrm{opt}}=16.1985\times10^{-9}~[s^3/rad^3]$ in this paper). 

For each Type-C WTG, the differential equations are given as follows:
\[
\begin{cases}         
\dot{E}_{qD}'=-\frac{1}{T_{0}'}(E_{qD}'+(X_{s}-X_{s}')I_{ds})+\omega_{s}\frac{X_{m}}{X_{r}}V_{dr}\\
~~~~~~~~-(\omega_{s}-\omega_{r})E_{dD}', \\
\dot{E}_{dD}'=-\frac{1}{T_{0}'}(E_{dD}'-(X_{s}-X_{s}')I_{qs})-\omega_{s}\frac{X_{m}}{X_{r}}V_{qr}      \\
~~~~~~~~+(\omega_{s}-\omega_{r})E_{qD}', \\
~~d\omega_{r}=\frac{\omega_{s}}{2H_{D}}(T_{m}-E_{dD}'I_{ds}-E_{qD}'I_{qs})dt+k_wdw,\\
~~\dot{x}_{1}=K_{I1}(P_{\mathrm{r}\mathrm{e}\mathrm{f}}-P_{\mathrm{g}\mathrm{e}\mathrm{n}}),                \\
~~\dot{x}_{2}=K_{I2}(K_{P1}(P_{\mathrm{r}\mathrm{e}\mathrm{f}}-P_{\mathrm{g}\mathrm{e}\mathrm{n}})+x_{1}-I_{qr}),\\
~~\dot{x}_{3}=K_{I3}(Q_{\mathrm{r}\mathrm{e}\mathrm{f}}-Q_{\mathrm{g}\mathrm{e}\mathrm{n}}),\\
~~\dot{x}_{4}=K_{I4}(K_{P3}(Q_{\mathrm{r}\mathrm{e}\mathrm{f}}-Q_{\mathrm{g}\mathrm{e}\mathrm{n}})+x_{3}-I_{dr}),
\end{cases}                 
\]
where $E_{dD}'$, $E_{qD}'$ and $\omega_{r}$ are the $d$, $q$ axis voltage and rotor speed of the WTG, respectively, $k_w$ is a positive factor corresponding to the stochastic part of the wind power generation, $x_{1}$ to $x_{4}$ are proportional-integral (PI) regulator induced states, $K_{I1}$, $K_{I2}$, $K_{I3}$, $K_{I4}$, $K_{P1}$, $K_{P2}$, $K_{P3}$, $K_{P4}$ are parameters of the PI regulator, $T_{m}$ is the mechanical torque generated by the wind, $V_{dr}, V_{qr}, I_{dr}, I_{qr}$ are the rotor $d$, $q$ axis voltage and current, respectively, $I_{ds}, I_{qs}$ are the stator $d$, $q$ axis current, respectively, $P_{\mathrm{g}\mathrm{e}\mathrm{n}}$ and $Q_{\mathrm{g}\mathrm{e}\mathrm{n}}$ are the WTG active and reactive power output, respectively, and

\begin{figure}[ht]
	\centering
	\includegraphics[scale=0.3]{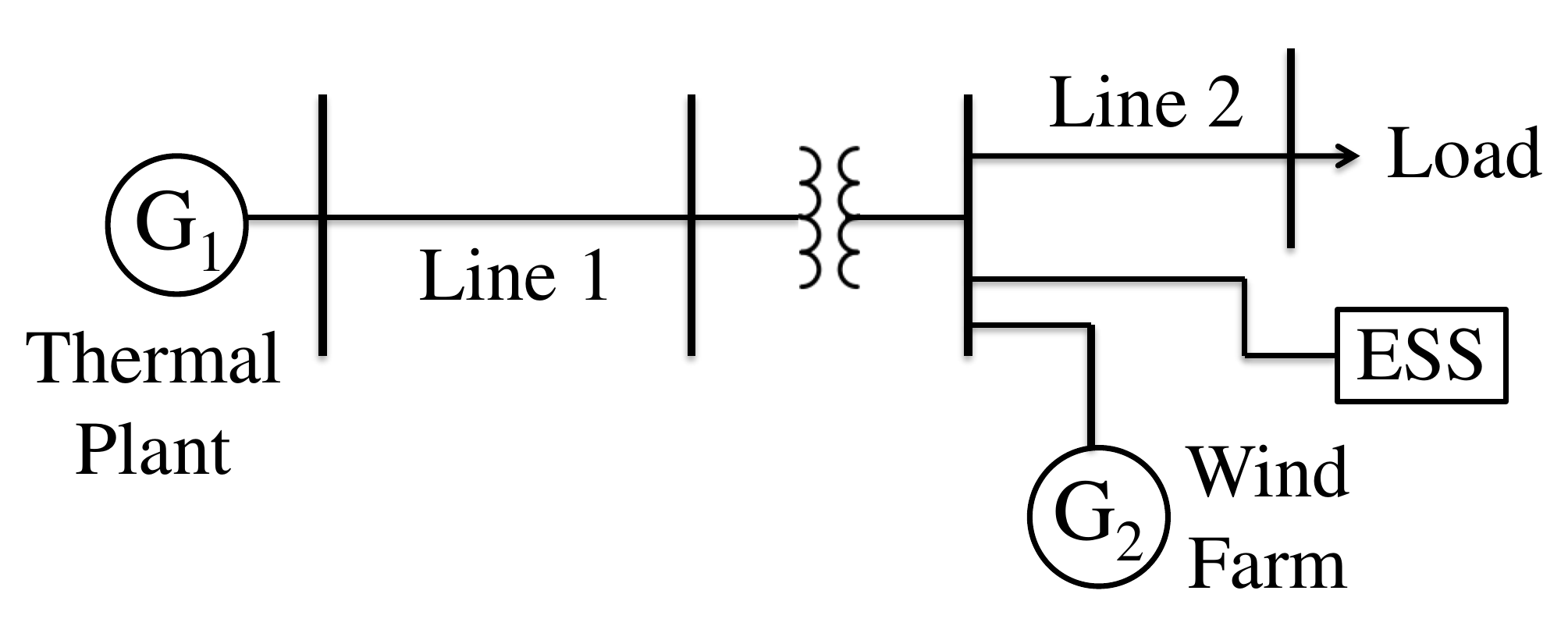}
	\caption{The four-bus system \cite{ZhangPulgar2017} with a thermal plant, a wind farm and an energy storage system (ESS).} 
	\label{dfig}
\end{figure}
\begin{align}\nonumber
	\begin{split}
	& P_{\mathrm{r}\mathrm{e}\mathrm{f}}=C_{\mathrm{o}\mathrm{p}\mathrm{t}}\omega_{r}^{3}+u^w, Q_{\mathrm{r}\mathrm{e}\mathrm{f}}=Q_{\mathrm{s}\mathrm{e}\mathrm{t}},\\
	& T_{0}'=\frac{X_{r}}{\omega_{s}R_{r}},~X_{s}'=X_{s}-\frac{X_{m}^{2}}{X_{r}}, \lambda=\frac{2k\omega_{r}R_{t}}{pv_{\mathrm{w}\mathrm{i}\mathrm{n}\mathrm{d}}}, \\
	&\lambda_{i}=(\frac{1}{\lambda+0.08\theta_{t}}-\frac{0.035}{\theta_{t}^{3}+1})^{-1}, \\  
    & C_{p}=0.22(\displaystyle \frac{116}{\lambda_{i}}-0.4\theta_{t}-5)e^{-\frac{12.5}{\lambda_{i}}}, \\		
	& T_{m}=\displaystyle \frac{1}{2}\frac{\rho\pi R_{t}^{2}\omega_{b}C_{p}v_{\mathrm{w}\mathrm{i}\mathrm{n}\mathrm{d}}^{3}}{S_{b}\omega_{r}},
	\end{split}  
\end{align}
where $u^w$ is a control input (by adjusting the input $u^w$, the wind turbine generator can adjust its power output for restoring the grid frequency to allowable ranges after a disturbance), the explanations of the other parameters can be found in Section 2.1.2 of \cite{PulgarPhd}.
We set $S_b$=1MVA, $P_{\mathrm{g}\mathrm{e}\mathrm{n}}=1.5$, $v_{\mathrm{w}\mathrm{i}\mathrm{n}\mathrm{d}}$=12m/s for the operating condition.

The algebraic equations of each Type-C WTG are given as follows:
\[
\begin{cases}
0=K_{P2}(K_{P1}(P_{\mathrm{r}\mathrm{e}\mathrm{f}}-P_{\mathrm{g}\mathrm{e}\mathrm{n}})+x_{1}-I_{qr})+x_{2}-V_{qr},                   \\
0=K_{P4}(K_{P3}(Q_{\mathrm{r}\mathrm{e}\mathrm{f}}-Q_{\mathrm{g}\mathrm{e}\mathrm{n}})+x_{3}-I_{dr})+x_{4}-V_{dr},\\
0=-P_{\mathrm{g}\mathrm{e}\mathrm{n}}+E_{dD}'I_{ds}+E_{qD}'I_{qs}-R_{s}(I_{ds}^{2}+I_{qs}^{2})\\
~~~~-(V_{qr}I_{qr}+V_{dr}I_{dr}),\\
0=-Q_{\mathrm{g}\mathrm{e}\mathrm{n}}+E_{qD}'I_{ds}-E_{dD}'I_{qs}-X_{s}'(I_{ds}^{2}+I_{qs}^{2}),\\
0=-I_{dr}+\displaystyle \frac{E_{qD}'}{X_{m}}+\frac{X_{m}}{X_{r}}I_{ds},\\
0=-I_{qr}-\displaystyle \frac{E_{dD}'}{X_{m}}+\frac{X_{m}}{X_{r}}I_{qs}.
\end{cases}                 
\]

The network algebraic equations are given as follows (details of the Type-C wind turbine generator network can be found in Figure 2.3 of \cite{PulgarPhd}):
\begin{center}
	$E_{qD}'-jE_{dD}'=(R_{s}+jX_{s}')(I_{qs}-jI_{ds})+V_{D}$,  	
	$V_{D}e^{j\theta_{D}}=jX_{t}(I_{qs}-jI_{ds}-I_{GC})e^{j\theta_{D}}+Ve^{j\theta}$,   
\end{center}
where $V_{D}$ and $\theta_{D}$ are voltage magnitude and angle of the bus to which the WTG is connected, and
$$
I_{GC}=\frac{V_{qr}I_{qr}+V_{dr}I_{dr}}{V_{D}}.
$$

By linearizing the system of differential-algebraic equations at the equilibrium point (the equilibrium point can be found by calculating the root of the algebraic equations and the right-hand side of the differential equations equal to zero), we have

	\begin{align}
	\begin{split}                                            
	& d\begin{bmatrix}
	\Delta x \\
       0\\		
	\end{bmatrix}=
	\begin{bmatrix}
	A_s& B_s \\
	C_s& D_s \\		
	\end{bmatrix}
	\begin{bmatrix}
	\Delta x \\
    \Delta y\\		
	\end{bmatrix}dt+
	\begin{bmatrix}
	 M_s \\
	 N_s\\		
	\end{bmatrix}u^wdt+\begin{bmatrix}
	\Sigma_{s1} \\
	\Sigma_{s2}\\		
	\end{bmatrix}dw,\\
	&\Delta P_{\rm{gen}}=
	\begin{bmatrix}
	E_s ~~
	F_s\\		
	\end{bmatrix}
	\begin{bmatrix}
	\Delta x \\
	\Delta y \\		
	\end{bmatrix},
	\end{split}
	\end{align}
where $\Delta x=[\Delta E_{qD}, \Delta E_{dD}, \Delta\omega_r, \Delta x_1, \Delta x_2, \Delta x_3, \Delta x_4]^T$, $\Delta y=[\Delta P_{\rm{gen}}, \Delta Q_{\rm{gen}}, \Delta V_{dr}, \Delta V_{qr}, \Delta I_{dr}, \Delta I_{qr}, \Delta I_{ds}, \Delta I_{qs},$ $\Delta V_D, \Delta \theta_D]^T$, $\triangle P_{\mathrm{gen}}$ is the active power variation from each WTG.	

Through the Kron reduction, we have
\begin{align}
\begin{split}
& d\Delta x=A_{\rm{kr}}\Delta xdt+B_{\rm{kr}}u^wdt+\Sigma_{\rm{kr}}dw,                   \\
&\Delta P_{\rm{gen}}=C_{\rm{kr}}\Delta x+D_{\rm{kr}}u^w+E_{\rm{kr}}dw/dt,
\end{split}
\label{Kron}
\end{align}                 
where 
\begin{center}
~~~~$A_{\rm{kr}}=A_s-B_sD_s^{-1}C_s$, ~~~
$B_{\rm{kr}}=M_s-B_sD_s^{-1}N_s$, ~~~
~~~~$C_{\rm{kr}}=E_s-F_sD_s^{-1}C_s$, ~~~
$D_{\rm{kr}}=-F_sD_s^{-1}N_s$,
~~~~$\Sigma_{\rm{kr}}=\Sigma_{s1}-B_sD_s^{-1}\Sigma_{s2}$, ~~
$E_{\rm{kr}}=-F_sD_s^{-1}\Sigma_{s2}$.
\end{center}

From 0 second to 5 seconds after the disturbance, the system frequency response model of the the four-bus system is as follows (we choose base MVA as 1000MVA):
\begin{align}
\begin{cases}
&\triangle\dot{\omega}=\frac{\omega_{s}}{2H}(\triangle P_{m}+u^s+\triangle P_s-\triangle P_d+200\Delta P_{\rm{gen}}/1000\\&~~~~~~~-\frac{D}{\omega_{s}}\triangle\omega),                    \\
&\triangle\dot{P}_s=0;\\
&\triangle\dot{P}_{m}=\frac{1}{\tau_{\mathrm{c}\mathrm{h}}}(\triangle P_{v}-\triangle P_{m}),\\
&\triangle\dot{P}_{v}=\frac{1}{\tau_{g}}(-\triangle P_{v}-\frac{1}{2\pi R}\triangle\omega),
\end{cases}
\label{SG1}
\end{align}                
where $\triangle\omega$ is the grid frequency deviation, $\triangle P_{m}$ is the governor mechanical power variation, $\triangle P_{v}$ is the governor valve position variation and $\triangle P_{d}$ denotes a large disturbance (e.g. generation loss or abrupt load changes). $\triangle P_{\mathrm{gen}}$ times 200 as there are 200 WTGs, and it is divided by 1000 as the base MVA for each WTG and the power system are 1 MVA and 1000 MVA, respectively. We set $\omega_{s}=2\pi\times60$rad/s, $D$=1, $H$=4s, $\tau_{\mathrm{ch}}$=0.3s, $\tau_{g}$=0.1s, $R$=0.05.  

From 5 seconds to 8.75 seconds after the disturbance, the system frequency response model is as follows:

\begin{align}
\begin{cases}
&\triangle\dot{\omega}=\frac{\omega_{s}}{2H}(\triangle P_{m}+u^s+\triangle P_s-\triangle P_d+200\Delta P_{\rm{gen}}/1000\\&~~~~~~~-\frac{D}{\omega_{s}}\triangle\omega),                    \\
&\triangle\dot{P}_s=0.04;\\
&\triangle\dot{P}_{m}=\frac{1}{\tau_{\mathrm{c}\mathrm{h}}}(\triangle P_{v}-\triangle P_{m}),\\
&\triangle\dot{P}_{v}=\frac{1}{\tau_{g}}(-\triangle P_{v}-\frac{1}{2\pi R}\triangle\omega),
\end{cases}
\label{SG2}
\end{align}     

At 8.75 seconds, the generation and load are balanced again. So from 8.75 seconds to 10 seconds after the disturbance, the system frequency response model is the same as that in (\ref{SG1}).

With (\ref{Kron}), (\ref{SG1}) and (\ref{SG2}), we have the following switched stochastic control system with two modes corresponding to (\ref{SG1}) and (\ref{SG2}) respectively:        
\begin{align}
d\hat{x}=(\hat{A}_q\hat{x}+\hat{B}_qu)dt+\hat{\Sigma}_qdw,
\label{whole}
\end{align}  
where $\hat{x}=[\triangle E_{qD}, \triangle E_{dD}, \triangle\omega_r, \triangle x_1, \triangle x_2, \triangle x_3, \triangle x_4,\triangle\omega,$ $\triangle P_s, \triangle P_{m},\triangle P_{v}]^T$, the input $u=[u^w, u^s]^T$. As the matrix $\hat{A}_q$ is computed as Hurwitz for both modes, the system in each mode is stable.                                                   

We consider a disturbance of generation loss of 150 MW (loss of
one unit, $\Delta P_{d}=0.15$). We use the following MTL specification for frequency regulation after the disturbance:	  
\begin{align}
\begin{split}
\varphi =& \Box_{[0,T_{\textrm{end}}]}p_1 \wedge \Box_{[2,T_{\textrm{end}}]}p_2,\\
p_1 =& (-0.5{\rm Hz}\leq \Delta f\leq 0.5{\rm Hz})\wedge(-10{\rm Hz}\leq \Delta f_{r}\leq 10{\rm Hz}),\\
p_2=& (-0.4{\rm Hz}\leq \Delta f\leq 0.4{\rm Hz}),
\end{split}
\label{spec}
\end{align}	
where $\Delta f=\frac{\triangle\omega}{2\pi}$, $\Delta f_{\rm{r}}=\frac{\triangle\omega_{\rm{r}}}{2\pi}$. The specification means ``After a disturbance, the grid frequency deviation should never exceed $\pm$0.5Hz, the WTG rotor speed deviation should never exceed $\pm$10Hz, after 2 seconds the grid frequency deviation should always be within $\pm$0.4Hz''.  
	
\begin{figure}
	\centering
	\includegraphics[scale=0.3]{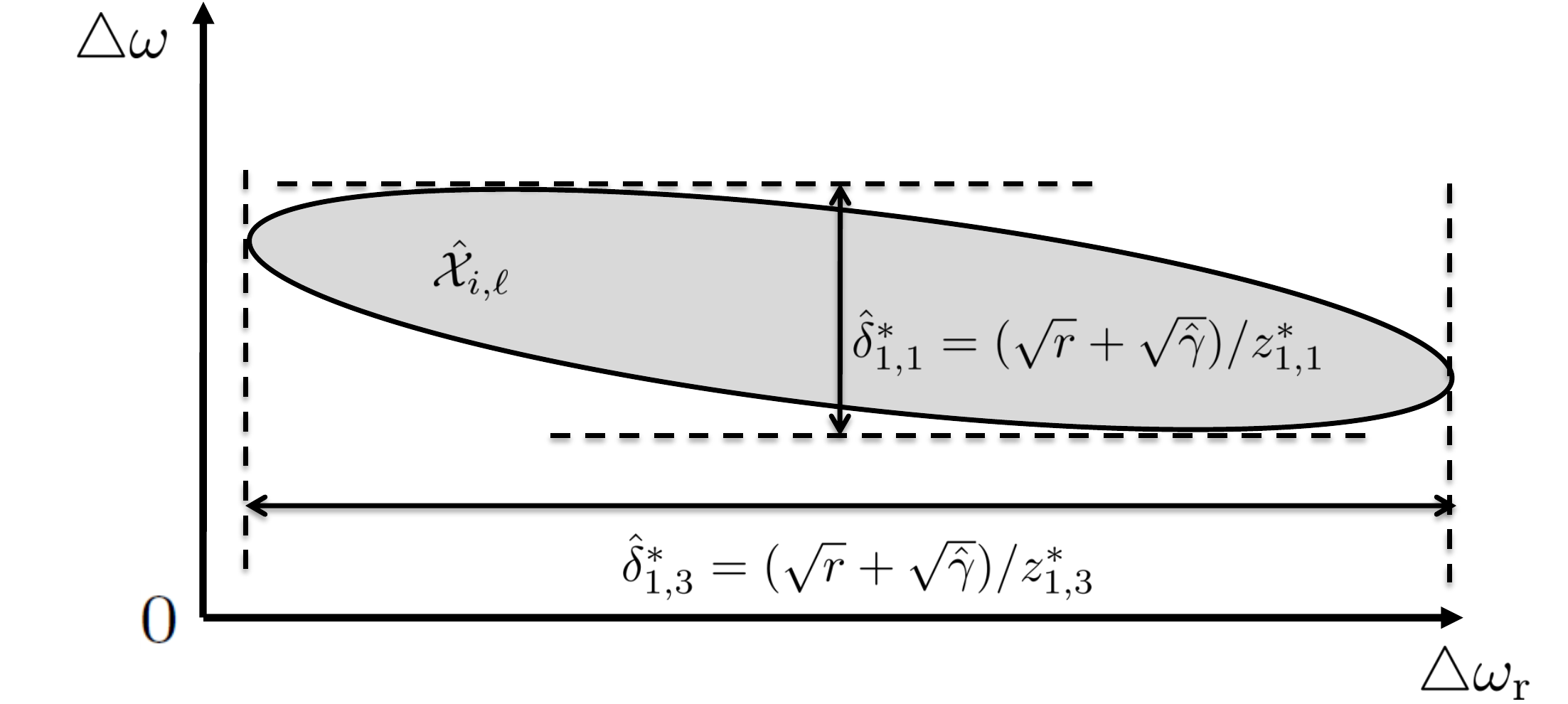}
	\caption{Outer bounds of states $\triangle\omega_{\rm{r}}, \triangle\omega$ in the stochastic robust neighbourhood $\hat{\mathcal{X}}_{i,\ell}$.} \label{ellipse}
\end{figure}

	\begin{table}[]
		\label{{parameter}}  
		\centering
		\caption{System parameters}
		\label{my-label}
		\begin{tabular}{lll}
			\toprule[2pt] 
			VA base $P_\textrm{b}$                 &       & 1000MVA \\ \hline
			System frequency $f_\textrm{s}$            &       & 60Hz   \\ \hline                        
			Active power flow to load $\textrm{L}_{i}$         & $i$=1       & 0.4 (pu)   \\ \hline
										& $i$=2       & 0.1 (pu)   \\ \hline
										& $i$=3       & 0.05 (pu)    \\ \hline
										& $i$=4       & 0.05 (pu)    \\ \hline
			Transformer impedance                    & $\textrm{G}_1$          & 1.8868 (pu)     \\ \hline                                   & $\textrm{G}_2$          & 0.618 (pu)    \\ \bottomrule[2pt] 
		\end{tabular}
	\end{table}     
	
	\begin{table}[]
		\label{{parameter}}  
		\centering
		\caption{Line data (1000 MVA base).}
		\label{my-label}
		\begin{tabular}{lll}
			\toprule[2pt]    
			Line number    & Line impedance (pu)   & Line charging (pu) \\ \hline
			2-8(2-9)          & j0.01 & 0.0006625\\ \hline
			7-8(7-9)          & j0.04 & 0.0023 \\ \hline
			4-8(4-9)          & j0.03 & 0.0031 \\ \hline
			4-5               & j0.03 & 0.0034  \\ \hline
			5-6               & j0.03 & 0.0094 \\ \hline
			6-7               & j0.02 & 0.0258 \\ \bottomrule[2pt]      
		\end{tabular}
	\end{table}

We set $k_w=1$, $T_{\textrm{end}}=5$ (s), $\epsilon=\alpha T_{\textrm{end}}/\hat{\gamma}=5\%$, so $\alpha=0.05\hat{\gamma}/T_{\textrm{end}}=0.01\hat{\gamma}$. As $\alpha=trace(\hat{\Sigma}^{T}M\hat{\Sigma})=k_w^2M(3,3)$, we have $\hat{\gamma}=100k_w^2M(3,3)=100M(3,3)$. We assume that the initial state variations can be covered by $B_{q^0}(\hat{x}^{\ast}_0,r)$, where $r=4\hat{\gamma}$ ($4=2^2$ is chosen as the initial state variations due to the time needed for running the algorithm to generate the controller, which is about twice the simulation time), $\hat{x}^{\ast}_0$ is zero in every dimension. It can be seen from (\ref{spec}) that the allowable variation range of the grid frequency variation $\triangle\omega$ is much smaller than that of the wind turbine rotor speed variation $\triangle\omega_{\rm{r}}$. Therefore, in order to decrease the conservativeness of the probabilistic bound as much as possible, we further optimize both $z_{k,i}$ and the matrix $M_{q^i}$ (geometrically change the shape of the stochastic robust neighbourhoods, see Fig. \ref{ellipse}) such that the outer bounds of the stochastic robust neighbourhoods in the dimension of the grid frequency variation $\hat{\delta}^{i\ast}_{1,1}$ ($\hat{\delta}^{i\ast}_{1,1}= \hat{\delta}^{i\ast}_{1,2}=\hat{\delta}^{i\ast}_{2,1}=\hat{\delta}^{i\ast}_{2,2}$) are much smaller than the outer bounds in the dimension of the wind turbine rotor speed variation $\hat{\delta}^{i\ast}_{1,3}$ ($\hat{\delta}^{i\ast}_{1,3}= \hat{\delta}^{i\ast}_{1,4}$). As $\hat{\delta}_{k,i}=(\sqrt{r_{q^i}}+\sqrt{\hat{\gamma}})/z_{k,i}$ and $\hat{\gamma}=100M_{q^0}(3,3)$, minimizing $\hat{\delta}_{1,1}$ can be achieved by minimizing $M_{q^i}(3,3)~(i=1,2,3)$ and maximizing $z^i_{1,1}$. The combined optimization to obtain both $M^{\ast}$ and $z^{\ast}_{1,1}$ is as follows:

\begin{align}
\begin{split}
&\textrm{min}. -(z^i_{1,1})^2\\	  
\textrm{s.t.} ~&  M_{q^i}\succ 0,
\hat{A}_{q^i}^{T}M_{q^i}+M_{q^i}\hat{A}_{q^i}+\mu M_{q^i}\preceq 0,\\
& e_3^{T}M_{q^i}e_3\leq \zeta, M_{q^i}-(z^i_{1,1})^2a_{1,1}a_{1,1}^{T}\succeq 0.
\end{split}
\label{opt1}
\end{align}
where $e_3=[0,0,1,0,0,0,0,0,0,0]^T$, $\mu=0.1$, $\zeta$ is tuned manually to be as small as possible while the optimization problem is feasible.

With the $M_{q^i}^{\ast}$ obtained from (\ref{opt1}), we compute the tightest outer bound in the dimension of $\triangle\omega_{\rm{r}}$ as follows:

\begin{align}
\begin{split}
&\textrm{min}. -(z^i_{1,3})^2\\	  
 \textrm{s.t.} ~  
& M_{q^i}^{\ast}-(z^i_{1,3})^2a_{1,3}a_{1,3}^{T}\succeq 0.
\end{split}
\label{opt2}
\end{align}

From (\ref{opt1}) and (\ref{opt2}), we obtain the $\hat{\delta}^{\ast}_{k,i}$-robust modified formula:
\[
\begin{split}
\varphi_{\hat{\delta}^{\ast}} =& \Box_{[0,T_{\textrm{end}}]}\hat{p}^{\ast}_1 \wedge \Box_{[2,T_{\textrm{end}}]}\hat{p}^{\ast}_2,\\
\hat{p}^{\ast}_1 =& (-0.5{\rm Hz}+0.217 e^{-0.01t}{\rm Hz}\leq \Delta f\\&\leq 0.5{\rm Hz}-0.217 e^{-0.01t}{\rm Hz})\wedge\\
&(-10{\rm Hz}+6.08 e^{-0.01t}{\rm Hz}\leq \Delta f_{r}\\&\leq 10{\rm Hz}-6.08 e^{-0.01t}{\rm Hz}),\\
\hat{p}^{\ast}_2=& (-0.4{\rm Hz}+0.217 e^{-0.01t}{\rm Hz}\leq \Delta f\\&\leq 0.4{\rm Hz}-0.217 e^{-0.01t}{\rm Hz}).
\end{split}
\]

\begin{figure}[th]
	\centering
	\includegraphics[width=8cm]{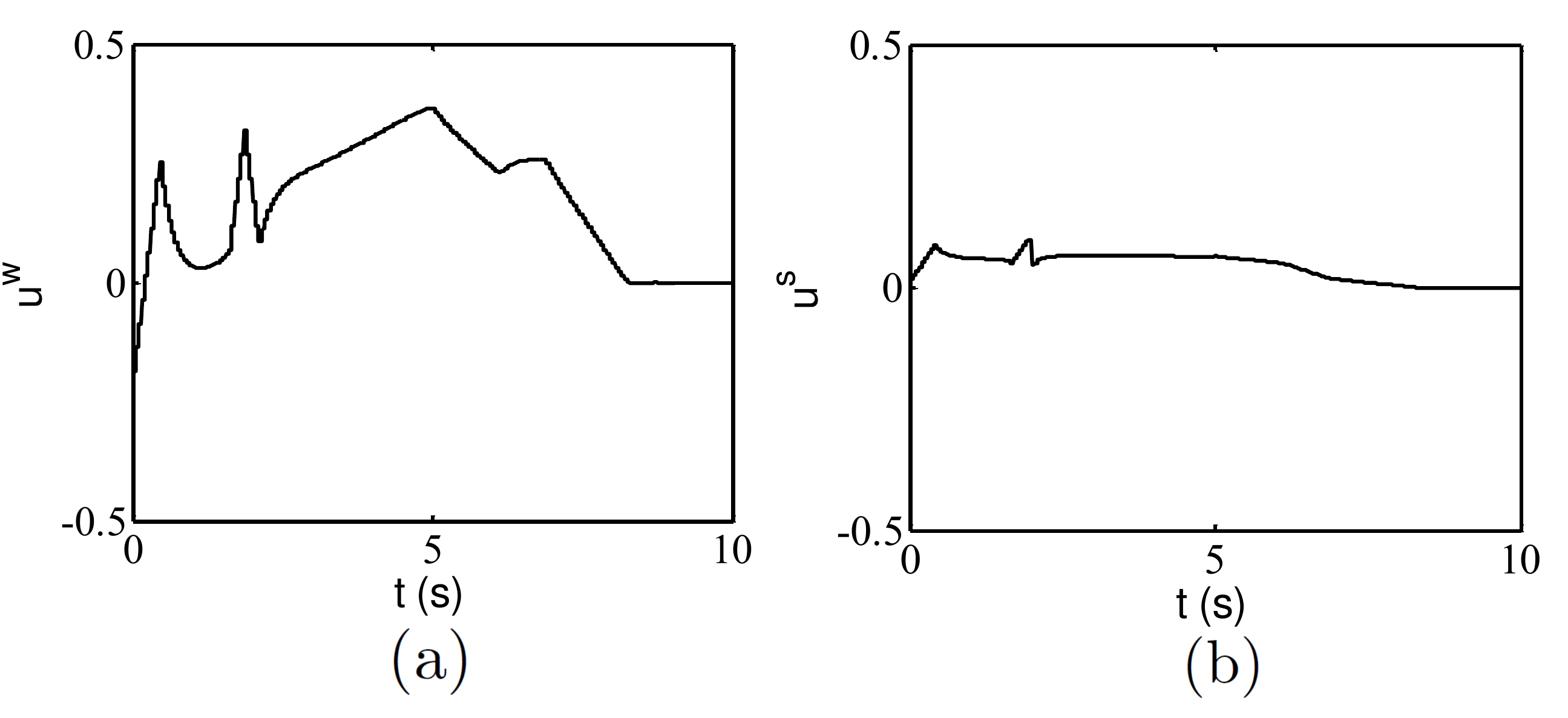}\caption{The obtained optimal input signals in case study I.}
	\label{wind_u}
\end{figure}          

\begin{figure}[th]
	\centering
	\includegraphics[width=8cm]{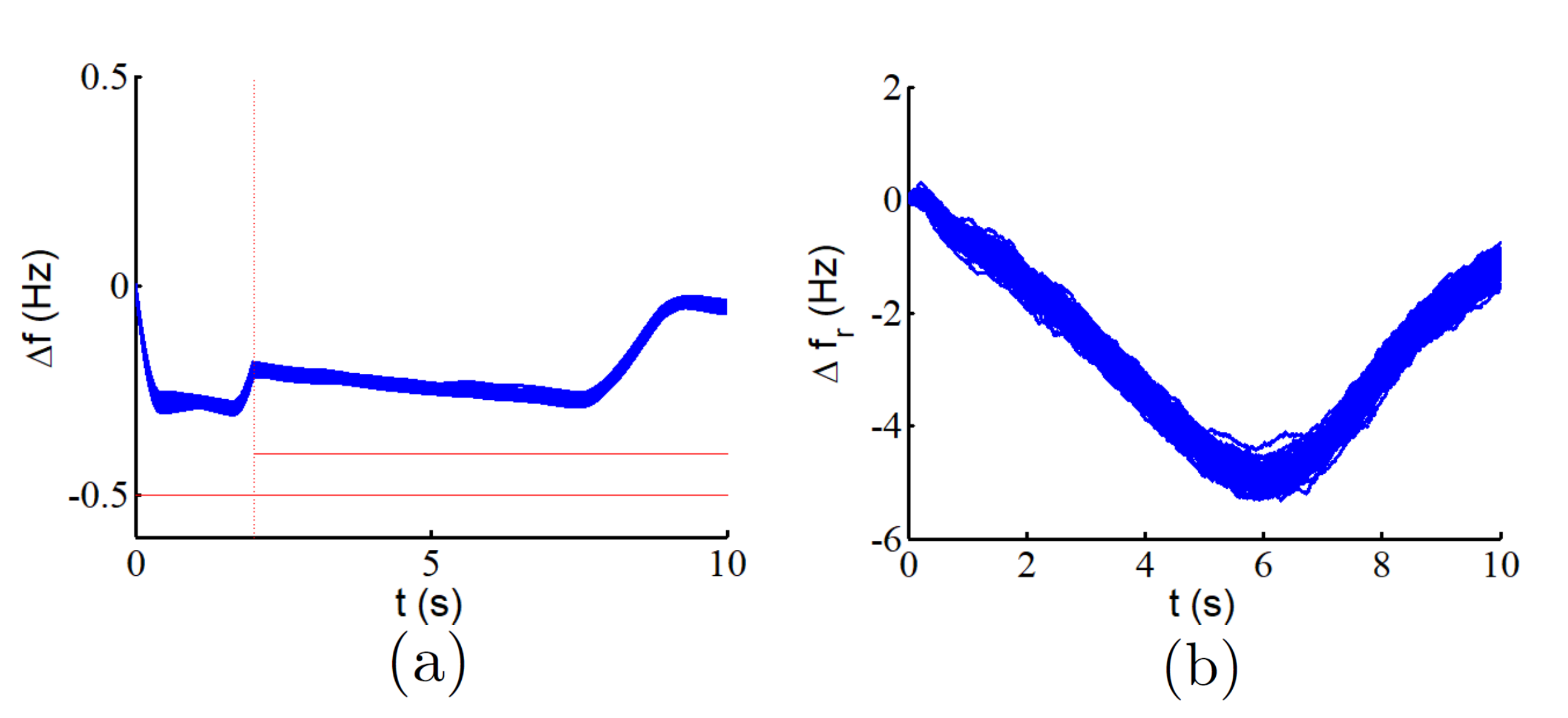}\caption{100 trajectories (realizations) of $\Delta f$ and $\Delta f_r$ without control (black) and with the feedforward controller (blue) in case study I.}
	\label{wind_f}
\end{figure}

We further design a feedback controller based on the obtained optimal input signals of the feedforward controller. To make a comparison between the feedback controller and the feedforward controller, we add an unexpected disturbance of per unit value 0.38 to $\Delta P_{d}$ during the first 0.1 second while generating 100 trajectories (realizations) of the stochastic system with both the feedforward and the feedback controller. As shown in Fig.~\ref{test_d}, 81\% of the trajectories generated with the feedforward controller still satisfy the MTL specification $\varphi$ with the minimal robustness degree of -0.0131, while all the trajectories generated with the feedback controller still satisfy the MTL specification $\varphi$ with the minimal robustness degree of 0.0235.
 
\begin{figure}[th]
	\centering
	\includegraphics[width=8cm]{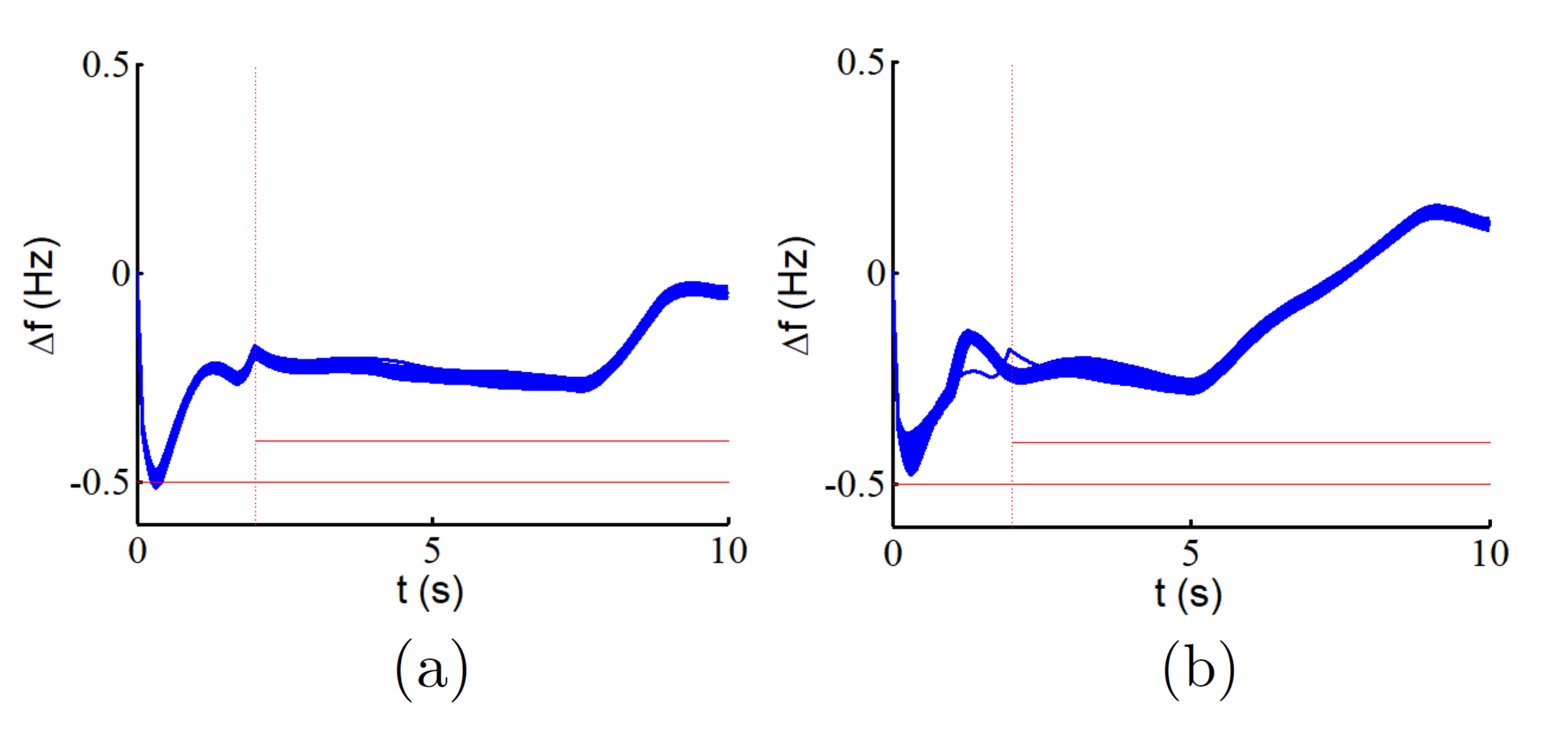}\caption{100 trajectories (realizations) of $\Delta f$ (a) with the feedforward controller and (b) with the feedback controller, with an unexpected disturbance of per unit value 0.38 to $\Delta P_{d}$ during the first 0.1 second.}
	\label{test_d}
\end{figure}

\subsection{Case Study II}
In this section, we apply the controller synthesis method on a nine-bus system as shown in Fig. \ref{gridnew}. The thermal plant $\textrm{G}_{1}$ and the wind farm $\textrm{G}_{2}$ are the same as those in Case Study I, with two energy storage systems (ESS) placed near them respectively. Four constant power loads are denoted as $\textrm{L}_1$, $\textrm{L}_2$, $\textrm{L}_3$ and $\textrm{L}_4$. The line data can be found in Tab. I and Tab. II \cite{zhe_control}.

\begin{figure}
	\centering
	\includegraphics[scale=0.3]{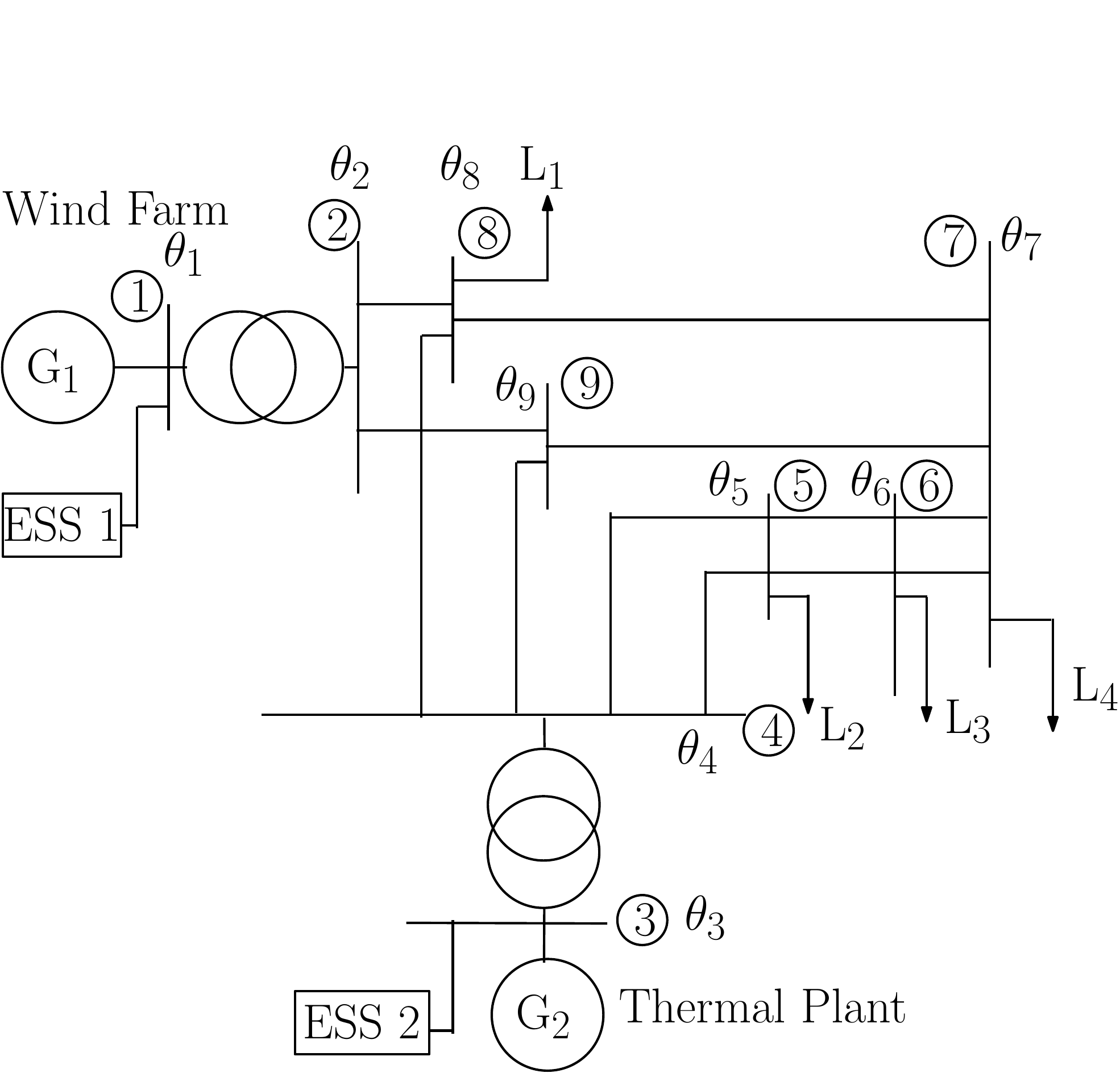}
	\caption{The nine-bus system with a thermal plant, a wind farm and two energy storage systems (ESS).} 
	\label{gridnew}
\end{figure}

We use the following MTL specification for frequency regulation after the disturbance:	  
\begin{align}
\begin{split}
\varphi =& \Box_{[0,T_{\textrm{end}}]}p_1 \wedge \Box_{[2,T_{\textrm{end}}]}p_2\wedge\Box_{[2,T_{\textrm{end}}]}p_3,\\
p_1 =& (-0.5{\rm Hz}\leq \Delta f\leq 0.5{\rm Hz})\wedge(-10{\rm Hz}\leq \Delta f_{r}\leq 10{\rm Hz}),\\
p_2=& (-0.4{\rm Hz}\leq \Delta f\leq 0.4{\rm Hz}),\\
p_3=&\bigwedge_{ij\in\mathcal{E}}(-0.25\leq P_{ij}\leq 0.25).
\end{split}
\label{spec2}
\end{align}

The first two subformulas in (\ref{spec2}) are the same as in (\ref{spec}) used in Section, while the third subformula $\Box_{[2,T_{\textrm{end}}]}p_3$ specifies the real power constraints in each line. We obtain the following $\hat{\delta}^{\ast}_{k,i}$-robust modified formula:
\[
\begin{split}
\varphi_{\hat{\delta}^{\ast}} =& \Box_{[0,T_{\textrm{end}}]}\hat{p}^{\ast}_1 \wedge \Box_{[2,T_{\textrm{end}}]}\hat{p}^{\ast}_2 \wedge\Box_{[2,T_{\textrm{end}}]}\hat{p}^{\ast}_3,\\
\hat{p}^{\ast}_1 =& (-0.5{\rm Hz}+0.217 e^{-0.01t}{\rm Hz}\leq \Delta f\\&\leq 0.5{\rm Hz}-0.217 e^{-0.01t}{\rm Hz})\wedge\\
&(-10{\rm Hz}+6.08 e^{-0.01t}{\rm Hz}\leq \Delta f_{r}\\&\leq 10{\rm Hz}-6.08 e^{-0.01t}{\rm Hz}),\\
\hat{p}^{\ast}_2=& (-0.4{\rm Hz}+0.217 e^{-0.01t}{\rm Hz}\leq \Delta f\\&\leq 0.4{\rm Hz}-0.217 e^{-0.01t}{\rm Hz}),\\
\hat{p}^{\ast}_3=&\bigwedge_{ij\in\mathcal{E}}(-0.25+0.0258 e^{-0.01t}\leq P_{ij}\\&\leq 0.25-0.0258 e^{-0.01t}).
\end{split}
\]
where $\mathcal{E}\subset\mathcal{N}\times\mathcal{N}$ is the set of transmission lines ($\mathcal{N}$ is the set of buses).

As there are 9 different lines corresponding to 9 different inequalities in the MTL specification, solving the optimization problem with all the inequality constraints could be computationally expensive. To reduce computation, we first set an initial MTL specification and iteratively add the line power flow inequality constraints that are violated with the previous optimization. The initial MTL specification $\varphi^0_{\hat{\delta}^{\ast}}$ as follows (by reducing the line power flow constraints in $\varphi_{\hat{\delta}^{\ast}}$):
\[
\begin{split}                                           
\varphi^0_{\hat{\delta}^{\ast}} =& \Box_{[0,T_{\textrm{end}}]}\hat{p}^{\ast}_1 \wedge \Box_{[2,T_{\textrm{end}}]}\hat{p}^{\ast}_2.
\end{split}
\]

We perform the feedforward controller synthesis with respect to $\varphi^0_{\hat{\delta}^{\ast}}$. We set $J(u(\cdot))=\norm{u^w(\cdot)}_2+\lambda\norm{u^s(\cdot)}_2$, where $\lambda=100$ (larger $\lambda$ encourages power input from the wind turbine generator). After the first iteration, the line 2-9 is overloaded and thus does not satisfy the line flow constraint in $\varphi_{\hat{\delta}^{\ast}}$. Thus we add line 2-9 power specification and obtain the following MTL specification $\varphi^1_{\hat{\delta}^{\ast}}$:
\[
\begin{split}
\varphi^1_{\hat{\delta}^{\ast}} =& \Box_{[0,T_{\textrm{end}}]}\hat{p}^{\ast}_1 \wedge \Box_{[2,T_{\textrm{end}}]}\hat{p}^{\ast}_2\wedge\Box_{[2,T_{\textrm{end}}]}\hat{p}^{\ast1}_3,\\
\hat{p}^{\ast1}_3=& (-0.25+0.0258 e^{-0.01t}\leq P_{28}\leq 0.25-0.0258 e^{-0.01t}).
\end{split}
\]
As shown in Fig. \ref{wind_f}, all of the 100 trajectories (realizations) starting from $B_{\psi}(\hat{x}^{\ast}_0,r)$ with the obtained optimal input signals  not only satisfy the MTL specification $\varphi^1_{\hat{\delta}^{\ast}}$, but already satisfy the original MTL specification $\varphi_{\hat{\delta}^{\ast}}$. Thus the iteration stops and the optimal input signals are obtained. The~obtained optimal input signals are shown in Fig. \ref{wind_u}.

\begin{figure}[th]                                          
	\centering
	\includegraphics[width=9cm]{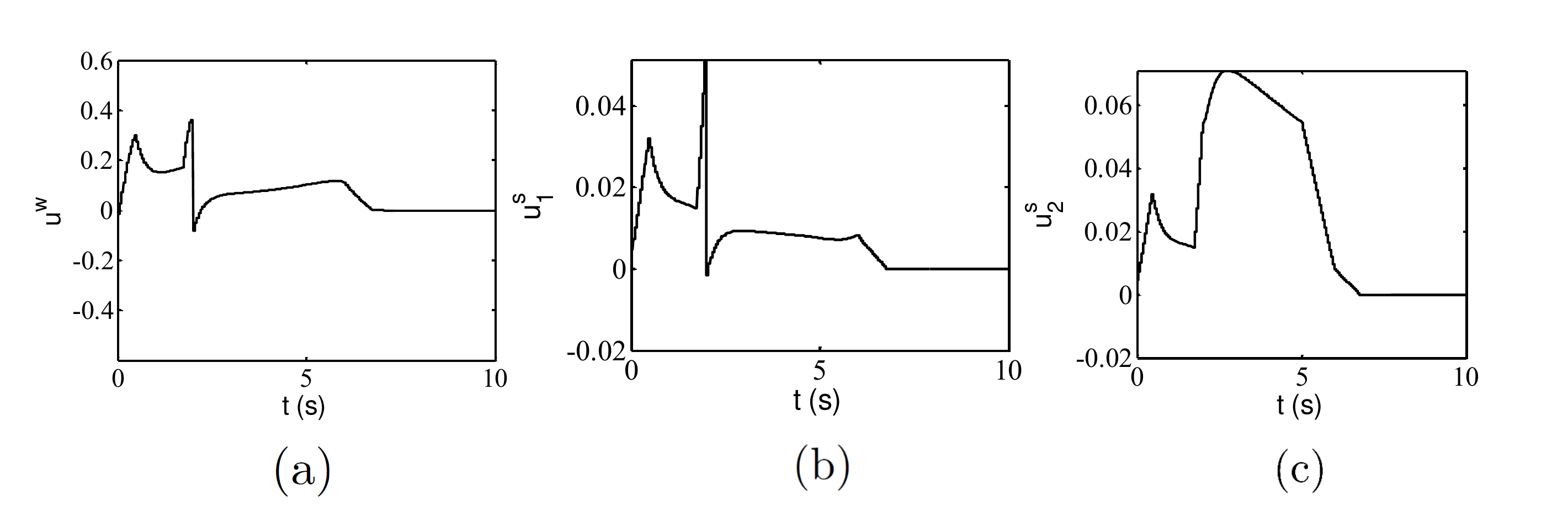}\caption{The obtained optimal input signals in case study II.}
	\label{wind_u}
\end{figure} 

We also design a feedback controller based on the obtained optimal input signals of the feedforward controller and the satisfaction rate is the same as that with the feedforward controller (100\%).  	

\begin{figure}[th]
	\centering
	\includegraphics[width=8cm]{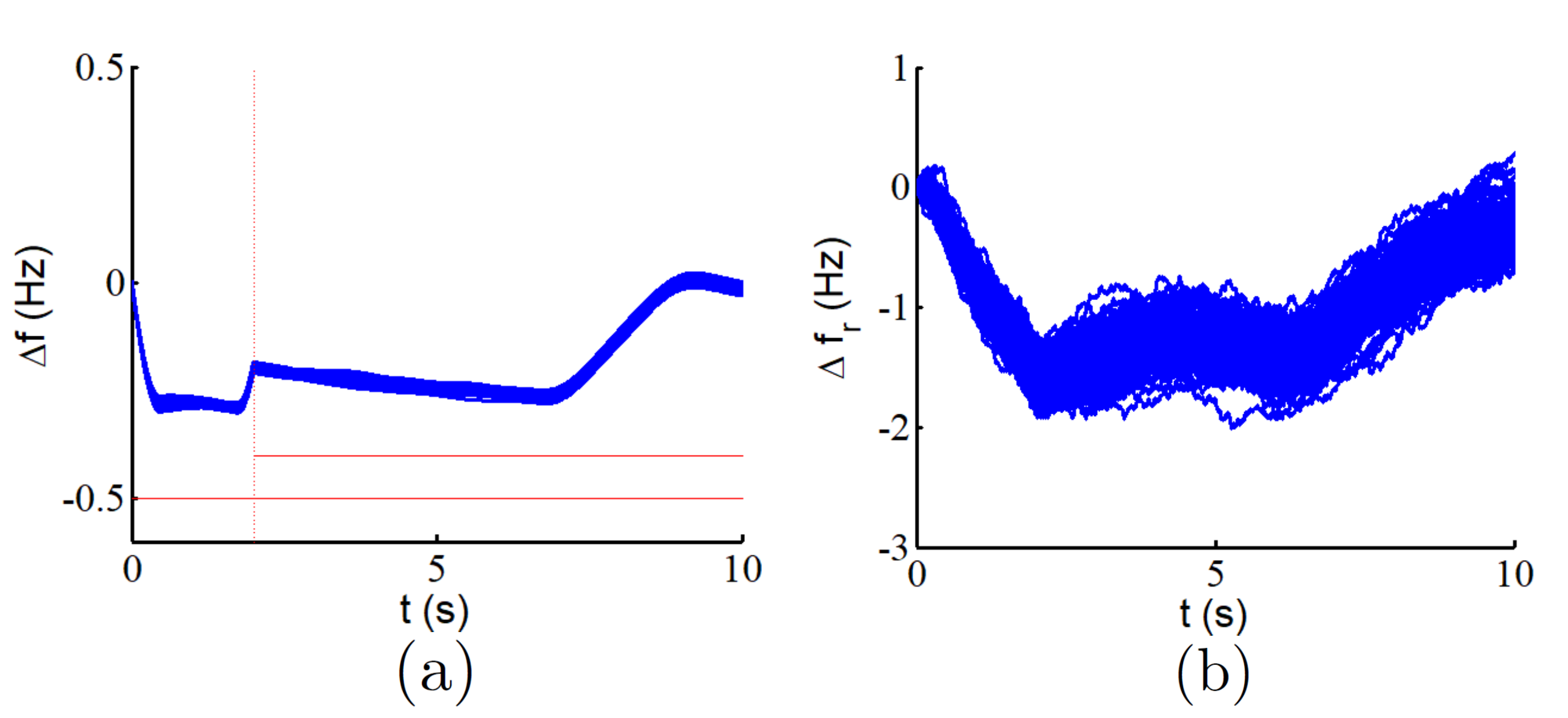}\caption{100 trajectories (realizations) of $\Delta f$ and $\Delta f_r$ with the feedforward controller (blue) in case study II.}
	\label{wind_f}
\end{figure}

\begin{figure}[th]
	\centering
	\includegraphics[width=8.5cm]{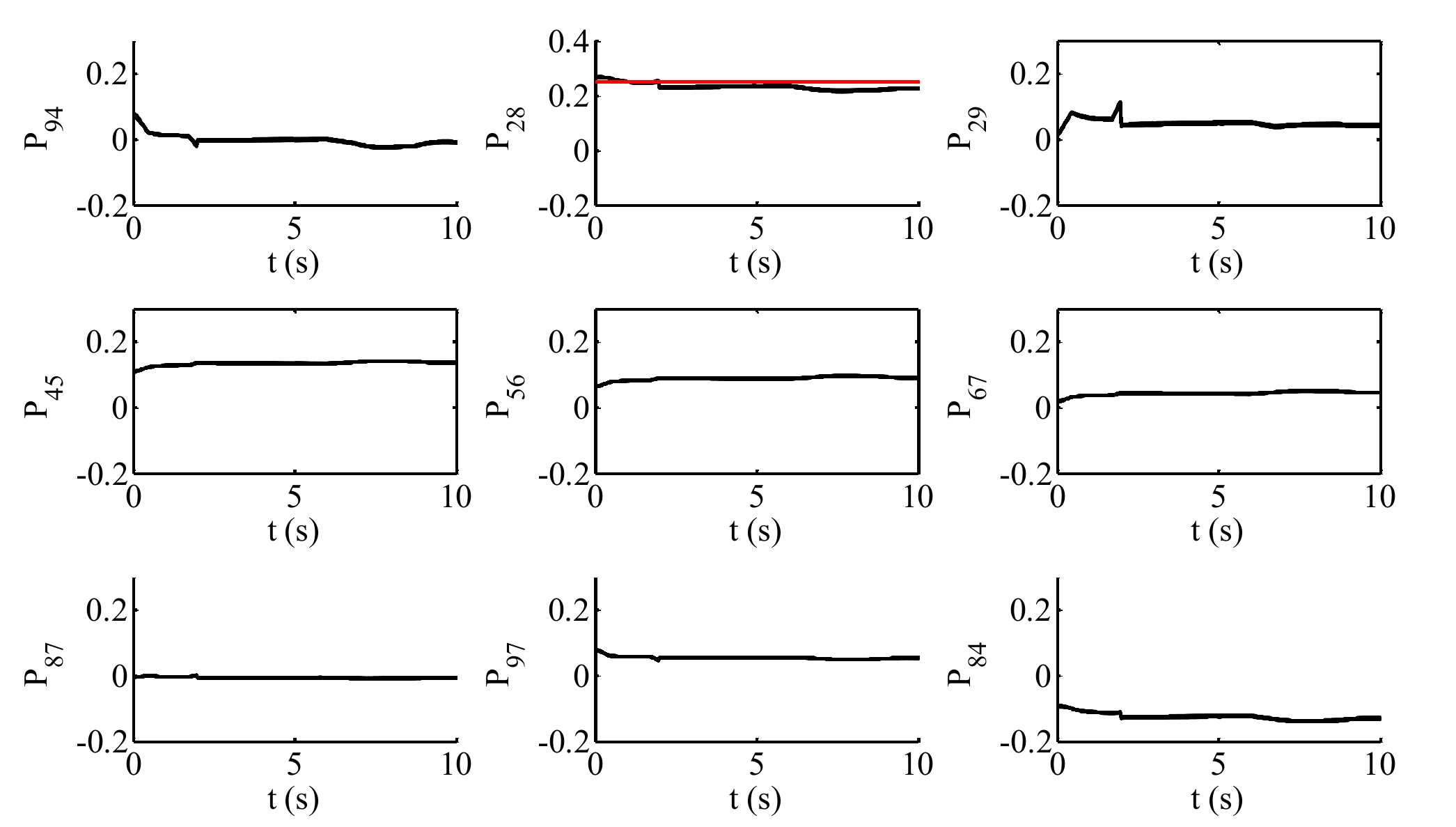}\caption{100 trajectories (realizations) of the real power of 9 different lines with the feedforward controller (blue) in case study II.}
	\label{wind_line}
\end{figure}

\section{CONCLUSIONS}
\label{conclusion}
We presented a coordinated control method of wind turbine generators and storage for frequency regulation under metric temporal logic specifications. Both the feedforward and feedback controllers are designed for this purpose. The controller synthesis approach can be also used in other related areas in power systems such as transient stability enhancement, voltage regulation, etc. Similar approaches can be applied in other switched control systems as well.

\section*{Acknowledgment}
This research was partially supported by the National Science
Foundation through grants CNS-1218109, CNS-1550029 and CNS-1618369 and the University of Tennessee. The authors would like to thank Dr. Yichen Zhang, Prof. Kevin Tomsovic, Prof. Hector Pulgar-Painemal, et al for introducing us to the WTG frequency support problem.

\section*{APPENDIX}
	\textbf{Proof of Theorem \ref{th1}}:\\
For the output trajectory $s_{\rho^{\ast}}({\bm\cdot};x^{\ast}_0,u)$ of trajectory $\rho^{\ast}=\{(q^{i},\xi^{\ast}_{q^{i}}(t;x^{\ast 0}_{q^i},u),T^{i})\}_{i=0}^{N_{q}}$ (where $x^{\ast 0}_{q^0}=x^{\ast}_0$) of the switched nominal control system, if $B_{q^{i-1}}(\xi_{q^{i-1}}(T^{i-1},x^{\ast 0}_{q^{i-1}}),r_{q^{i-1}}e^{-\mu^{i-1} T^{i-1}/2})\subset B_{q^{i}}(x^{\ast 0}_{q^i},$ $r_{q^{i}}) (i=1,2,\dots,N_q)$, then for any $\tilde{x}_0\in B_{q^0}(x^{\ast}_0,r_{q^0})$ and the output trajectory $s_{\tilde{\rho}^{\ast}}({\bm\cdot};\tilde{x}_0,u)$ of trajectory $\tilde{\rho}^{\ast}=\{(q^{i},\xi^{\ast}_{q^{i}}(t;\tilde{x}^{\ast 0}_{q^i},u),T^{i})\}_{i=0}^{N_{q}}$ (where $\tilde{x}^{\ast 0}_{q^0}=\tilde{x}_0$) of the switched nominal control system, we have $\tilde{x}^{\ast0}_{q^i}\in B_{q^i}(x^{\ast 0}_{q^i},r_{q^i})$.

For every $k\in\{1,\dots,\eta\}$, $\nu\in\{1,\dots,n_k\}$ and any $\tilde{x}^{\ast 0}_{q^i}\in B_{q^i}(x^{\ast 0}_{q^i},r_{q^i})$, we have
	\begin{align}
	\begin{split}
	& (\xi_{q^i}^{\ast}(t;\tilde{x}^{\ast 0}_{q^i},u)-\xi_{q^i}^{\ast}(t;x^{\ast 0}_{q^i},u))^{T}a_{k,\nu}a_{k,\nu}^{T}(z_{k,\nu}^{i})^2(\xi_{q^i}^{\ast}(t;\tilde{x}^{\ast 0}_{q^i},u)\\
	&-\xi_{q^i}^{\ast}(t; x^{\ast 0}_{q^i},u))\le (\xi_{q^i}^{\ast}(t;\tilde{x}^{\ast 0}_{q^i},u)-\xi_{q^i}^{\ast}(t;x^{\ast 0}_{q^i},u))^{T}M_{q^i}\\
	&(\xi_{q^i}^{\ast}(t;\tilde{x}^{\ast 0}_{q^i},u)-\xi_{q^i}^{\ast}(t;x^{\ast 0}_{q^i},u))\\&=\psi_{q^i}(\xi_{q^i}^{\ast}(t;\tilde{x}^{\ast 0}_{q^i},u), \xi_{q^i}^{\ast}(t;x^{\ast 0}_{q^i},u))e^{-\mu_{q^i} t}\le r_{q^i}e^{-\mu_{q^i} t}.
	\end{split}  
	\end{align}		
	Therefore, we have $\norm{a_{k,\nu}^{T}(\xi_{q^i}^{\ast}(t;\tilde{x}^{\ast 0}_{q^i},u)-\xi_{q^i}^{\ast}(t;x^{\ast 0}_{q^i},u))}\le\sqrt{r_{q^i}}e^{-\mu_{q^i} t/2}/z^i_{k,\nu}$, thus 
	\begin{align}
	  \begin{split}
	& -\sqrt{r_{q^i}}e^{-\mu_{q^i}t/2}/z^i_{k,\nu}\le a_{k,\nu}^{T}(\xi_{q^i}^{\ast}(t;\tilde{x}^{\ast 0}_{q^i},u)-\xi_{q^i}^{\ast}(t;x^{\ast 0}_{q^i},u))\\&\le\sqrt{r_{q^i}}e^{-\mu_{q^i} t/2}/z^i_{k,\nu}. 
       \end{split} 
	 \label{b1} 
	\end{align}	
	For every $k\in\{1,\dots,\eta\}$, $\nu\in\{1,\dots,n_k\}$ and output trajectory $s_{\tilde{\rho}}({\bm\cdot};\tilde{x}_0,u)$ of trajectory $\tilde{\rho}=\{(q^{i},\xi_{q^{i}}(t;\tilde{x}^{0}_{q^i},u),T^{i})\}_{i=0}^{N_{q}}$ (where $\tilde{x}^{0}_{q^0}=\tilde{x}^{\ast 0}_{q^0}=\tilde{x}_0$) of the switched stochastic control system, if $\sup_{0\leq t\leq T^i}\phi_{q^i}(\xi_{q^i}^{\ast}(t;\tilde{x}^{\ast0}_{q^i},u), \xi_{q^i}(t;\tilde{x}^0_{q^i},u))<\hat{\gamma}$, then $\xi_{q^i}(t;\tilde{x}^0_{q^i},u)\in B_{q^i}(\xi_{q^i}^{\ast}(t;\tilde{x}^{\ast 0}_{q^i},u),\hat{\gamma}e^{-\mu_{q^i} t/2})$, we have
	\begin{align}
	\begin{split}
	& (\xi_{q^i}(t;\tilde{x}^0_{q^i},u)-\xi_{q^i}^{\ast}(t;\tilde{x}^{\ast 0}_{q^i},u))^{T}a_{k,\nu}a_{k,\nu}^{T}(z_{k,\nu}^{i})^2(\xi_{q^i}(t;\tilde{x}^0_{q^i},u)\\
	&-\xi_{q^i}^{\ast}(t;\tilde{x}^{\ast 0}_{q^i},u))\le (\xi_{q^i}(t;\tilde{x}^0_{q^i},u)-\xi_{q^i}^{\ast}(t;\tilde{x}^{\ast 0}_{q^i},u))^{T}M_{q^i}\\
	&(\xi_{q^i}(t;\tilde{x}^0_{q^i},u)-\xi_{q^i}^{\ast}(t;\tilde{x}^{\ast 0}_{q^i},u))\\&=\phi_{q^i}(\xi_{q^i}(t;\tilde{x}^0_{q^i},u), \xi_{q^i}^{\ast}(t;\tilde{x}^{\ast 0}_{q^i},u))e^{-\mu_{q^i} t}\le\hat{\gamma}e^{-\mu_{q^i} t}.
	\end{split}  
	\end{align}
	Therefore, we have $\norm{a_{k,\nu}^{T}(\xi_{q^i}(t;\tilde{x}^0_{q^i},u)-\xi_{q^i}^{\ast}(t;\tilde{x}^{\ast 0}_{q^i},u))}\le\sqrt{\hat{\gamma}}e^{-\mu_{q^i} t/2}/z^i_{k,\nu}$, thus 
	\begin{align}
	\begin{split}
	& -\sqrt{\hat{\gamma}}e^{-\mu_{q^i} t/2}/z_{k,\nu}\le a_{k,\nu}^{T}(\xi_{q^i}(t;\tilde{x}^0_{q^i},u)-\xi_{q^i}^{\ast}(t;\tilde{x}^{\ast 0}_{q^i},u))\\&\le\sqrt{\hat{\gamma}}e^{-\mu_{q^i} t/2}/z^i_{k,\nu}.  
	\label{b2}
	\end{split} 
	\end{align}
	From (\ref{b1}) and (\ref{b2}), we have
	\begin{align}
	\begin{split}
	& -(\sqrt{\hat{\gamma}}+\sqrt{r_{q^i}})e^{-\mu_{q^i} t/2}/z^i_{k,\nu}\le a_{k,\nu}^{T}(\xi_{q^i}(t;\tilde{x}^0_{q^i},u)-\\&\xi_{q^i}^{\ast}(t;x^{\ast 0}_{q^i},u))\le(\sqrt{\hat{\gamma}}+\sqrt{r_{q^i}})e^{-\mu_{q^i} t/2}/z_{k,\nu}. 
	\end{split} 
	\label{b3}
	\end{align}		 	 
	
If $\left[\left[\varphi_{\hat{\delta}}\right]\right](s_{\rho^\ast}({\bm\cdot};x^{\ast}_0,u), 0)\ge 0$, where $\varphi_{\hat{\delta}}$ is the $\hat{\delta}_{k,\nu}$-robust modified formula of $\varphi$, $\hat{\delta}^i_{k,\nu}=(\sqrt{\hat{\gamma}}+\sqrt{r_{q^i}})/z^i_{k,\nu}$, then for every $k\in\{1,\dots,\eta\}$, $\nu\in\{1,\dots,n_k\}$, $i\in\{1,2,\dots,N_q\}$ (resp. $i=0$), and for any $t$ such that $t+\sum\limits_{j=1}^{i-1}T^j\ge\tau_k$ (resp. $t\ge\tau_k$ when $i=0$), we have $a_{k,\nu}^{T}\xi_{q^i}^{\ast}(t;x^{\ast}_0,u)+c_{k,\nu}^{T}u<b_{k,\nu}-(\sqrt{\hat{\gamma}}+\sqrt{r_{q^i}})e^{-\mu_{q^i} t/2}/z^i_{k,\nu}$. In such conditions, for any $\tilde{x}_0\in B_{q^0}(x^{\ast}_0,r_{q^0})$ (thus $\tilde{x}^{\ast0}_{q^i}\in B_{q^i}(x^{\ast 0}_{q^i},r_{q^i})$), if $\xi_{q^i}(t;\tilde{x}^0_{q^i},u)\in B_{q^i}(\xi_{q^i}^{\ast}(t;\tilde{x}^{\ast 0}_{q^i},u),\hat{\gamma}e^{-\mu_{q^i} t/2})$, we have
	\begin{align}
	\begin{split}\nonumber 
	& a_{k,\nu}^{T}\xi_{q^i}(t;\tilde{x}^0_{q^i},u)+c_{k,\nu}^{T}u< a_{k,\nu}^{T}\xi_{q^i}^{\ast}(t;x^{\ast 0}_{q^i},u)+c_{k,\nu}^{T}u+(\sqrt{\hat{\gamma}}\\&+\sqrt{r_{q^i}})e^{-\mu_{q^i} t/2}/z^i_{k,\nu}<b_{k,\nu}.
	\end{split}
	\end{align}     
Therefore, from the above analysis and (\ref{prob}), for any $\tilde{x}_0\in B_{q^0}(x^{\ast}_0,r_{q^0})$ we have ($0\le t\le T^i$ in the following notations)
	\begin{align}\nonumber   
	\begin{split}
	& P\{\left[\left[\varphi\right]\right](s_{\tilde{\rho}}({\bm\cdot};\tilde{x}_0,u), 0)\ge 0~\vert~\left[\left[\varphi_{\hat{\delta}}\right]\right](s_{\rho^\ast}({\bm\cdot};x^{\ast}_0,u), 0)\ge 0\}\\
	& \ge P\{\forall k, \forall \nu, \forall i,\forall t~\textrm{such~that}~t+\sum\limits_{j=1}^{i-1}T^j\ge\tau_k (\textrm{resp.}~ t\ge\tau_k ~\\&\textrm{when}~ i=0), a_{k,\nu}^{T}\xi_{q^i}(t;\tilde{x}^0_{q^i},u)+c_{k,\nu}^{T}u<b_{k,\nu}~\vert~\left[\left[\varphi_{\hat{\delta}}\right]\right]\\&~~ (s_{\rho^\ast}({\bm\cdot};x^{\ast}_0,u),0)\ge 0\}
	\\&\ge P\{\forall k, \forall \nu, \forall i,\forall t, \norm{a_{k,\nu}^{T}(\xi_{q^i}(t;\tilde{x}^0_{q^i},u)-\xi_{q^i}^{\ast}(t;x^{\ast 0}_{q^i},u))}<\\&~~~(\sqrt{\hat{\gamma}}+\sqrt{r_{q^i}})e^{-\mu_{q^i} t/2}/z_{k,\nu}~\vert\left[\left[\varphi_{\hat{\delta}}\right]\right](s_{\rho^\ast}({\bm\cdot};x^{\ast}_0,u), 0)\ge 0\}
	\\&= P\{\forall k, \forall \nu,\forall i,\forall t,\norm{a_{k,\nu}^{T}(\xi_{q^i}(t;\tilde{x}^0_{q^i},u)-\xi_{q^i}^{\ast}(t;x^{\ast 0}_{q^i},u))}<\\&~~~(\sqrt{\hat{\gamma}}+\sqrt{r_{q^i}})e^{-\mu_{q^i} t/2}/z_{k,\nu}\} 
	\\&\ge P\{\forall k, \forall \nu,\forall i,\forall t,\norm{a_{k,\nu}^{T}(\xi_{q^i}(t;\tilde{x}^0_{q^i},u)-\xi_{q^i}^{\ast}(t;\tilde{x}^{\ast 0}_{q^i},u))}<\\&~~~\sqrt{\hat{\gamma}}e^{-\mu_{q^i} t/2}/z_{k,\nu}\}
	\\&\ge P\{\sup_{0\leq t\leq T^0}\phi_{q^0}(\xi_{q^0}^{\ast}(t;\tilde{x}^{\ast0}_{q^0},u), \xi_{q^0}(t;\tilde{x}^0_{q^0},u))<\hat{\gamma}\}\times\dots\\& P\{\sup_{0\leq t\leq T^{N_q}}\phi_{q^{N_q}}(\xi_{q^{N_q}}^{\ast}(t;\tilde{x}^{\ast0}_{q^{N_q}},u), \xi_{q^{N_q}}(t;\tilde{x}^0_{q^{N_q}},u))<\hat{\gamma}\}\\&\ge(1-\frac{\alpha_{q^0} T^0}{\hat{\gamma}})\times(1-\frac{\alpha_{q^1} T^1}{\hat{\gamma}})\times\dots
	\times(1-\frac{\alpha_{q^{N_q}}T^{N_q}}{\hat{\gamma}})\\&\stackrel{\mathclap{\normalfont\mbox{(a)}}}{\ge} 1-\frac{\alpha_{q^0}T^0+\alpha_{q^1}T^1+\dots \alpha_{q^{N_q}}T^{N_q}}{\hat{\gamma}}\\&\ge 1-\frac{(\max\limits_{i}\alpha_{q^i})\cdot (T^0+T^1+\dots T^{N_q})}{\hat{\gamma}}\\&=1-\frac{(\max\limits_{i}\alpha_{q^i})\cdot T_{\textrm{end}}}{\hat{\gamma}}=1-\epsilon.    
	\end{split}  
	\end{align}             
The inequality $(a)$ follows from the fact that $(1-c_1)(1-c_2)\dots(1-c_n)\ge 1-(c_1+c_2+\dots+c_n)$ when $0\le c_i\le 1$ $(i=1,2,\dots,n)$, which can be easily proven by induction.  
                                                     
\textbf{Proof of Theorem \ref{th2}}:\\
We denote $\ell_0\triangleq\min\{\ell\vert \tilde{x}_0\in B_{q^0}(x^{\ast}_{0,{\ell}},r)\}$, so $\chi_u(x,t[0])=$ $u_{\ell_0}[0]$. If $\phi_{q^{1}}(\xi^{\ast}_{q^1,\ell_0}(t[1];x_0,u_{\ell_0}), \xi_{q^1,\ell_0}(t[1];x_0,u_{\ell_0}))<\hat{\gamma}$, then according to Algorithm \ref{alg}, $\chi_u(x,q^1,t[1])=$ $u_{\ell_0}[1]$, and so forth. So if $\forall i, \sup_{0\leq t[j^{q^i}_t]\leq T^{q^i}}\phi_{q^{i}}(\xi^{\ast}_{q^i,\ell_0}(j^{q^i}_t;\tilde{x}_0,u_{\ell_0}),$ $ \xi_{q^i,\ell_0}(j^{q^i}_t;\tilde{x}_0,u_{\ell_0}))<\hat{\gamma}$, then $\forall i, \forall j^{q^i}_t$, $\chi_u(x,q^i,t[j^{q^i}_t])=$ $u_{\ell_0}[j^{q^i}_t]$. Then as for any $\ell$, $\left[\left[\varphi_{\hat{\delta}^{\ast}}\right]\right](s_{\rho^{\ast}}({\bm\cdot};x^{\ast}_{0,\ell},u_{\ell}), 0)\ge 0$, we have $\left[\left[\varphi\right]\right]$ $(s_{\rho}({\bm\cdot};\tilde{x}_{0},\chi_u), 0)=\left[\left[\varphi\right]\right](s_{\rho}({\bm\cdot};\tilde{x}_{0},u_{\ell_0}), 0)\ge 0$. Therefore, we have 
	\begin{align}\nonumber   
	\begin{split} 
	& P\{\left[\left[\varphi\right]\right](s_{\rho}({\bm\cdot};\tilde{x}_{0},\chi_u), 0)\ge 0~\vert~\left[\left[\varphi_{\hat{\delta}^{\ast}}\right]\right](s_{\rho}^{\ast}({\bm\cdot};x^{\ast}_{0,\ell},u_{\ell}), 0)
	\\&~~~~\ge 0, \forall \ell\}\\&\ge P\{\sup_{0\leq t[j^{q^0}_t]\leq T^0}\phi_{q^0}(\xi_{q^0}^{\ast}(t[j^{q^0}_t];\tilde{x}^{\ast0}_{q^0},u_{\ell_0}), \xi_{q^0}(t[j^{q^0}_t];\tilde{x}^0_{q^0},u_{\ell_0}))<\hat{\gamma}\}\\&\times\dots P\{\sup_{0\leq t[j^{q^{N_q}}_t]\leq T^{N_q}}\phi_{q^{N_q}}(\xi_{q^{N_q}}^{\ast}(t[j^{N_q}_t];\tilde{x}^{\ast0}_{q^{N_q}},u_{\ell_0}),\\& \xi_{q^{N_q}}(t[j^{q^{N_q}}_t];\tilde{x}^0_{q^{N_q}},u_{\ell_0}))<\hat{\gamma}\}.
	\end{split}  
	\end{align}  

When there is no unexpected disturbances, we have
		\begin{align}\nonumber   
		\begin{split} 
		& P\{\sup_{0\leq t[j^q_t]\leq T^i}\phi_{q^i}(\xi_{q^i}^{\ast}(t[j^q_t];\tilde{x}^{\ast0}_{q^i},u_{\ell_0}), \xi_{q^i}(t[j^q_t];\tilde{x}^0_{q^i},u_{\ell_0}))<\hat{\gamma}\}\\&> 1-\frac{\alpha_{q^i} T^i}{\hat{\gamma}}.
		\end{split}  
		\end{align}  
		
So from the proof of Theorem \ref{th1}, $P\{\left[\left[\varphi\right]\right](s_{\rho}({\bm\cdot};\tilde{x}_{0},\chi_u), 0)\ge 0~\vert~\left[\left[\varphi_{\hat{\delta}^{\ast}}\right]\right](s_{\rho^{\ast}}({\bm\cdot};x^{\ast}_{0,\ell},u_{\ell}),$ $ 0)\ge 0,\forall \ell\}> 1-\epsilon$. 

\textbf{Proof of Theorem \ref{th3}}:\\	
We denote $\ell_0\triangleq\min\{\ell\vert \tilde{x}_0\in B_{q^0}(x^{\ast}_{0,{\ell}},r)\}$. We denote the $\tau$-time shifted input signal of $u_{\ell}(\cdot)$ as $R^\tau u_{\ell}(\cdot)$, i.e. $\forall t\ge0, R^\tau u_{\ell}(t)=u_{\ell}(t+\tau)$. We assume that in each mode $q$, at time instant $j_t^{q,1}$, $j_t^{q,2}$, $\dots$, $j_t^{q,m_q}$, the corresponding state is perturbed to another state $x^{q,1}$, $x^{q,2}$, $\dots$, $x^{q,m_q}$, respectively, where $x^k\in\bigcup\limits_{\ell=1}^{N}B_{q}(\xi^{\ast}_{q,\ell}[i_t^k],r_q)$, $j_t^k\le i_t^k\le j_t^k+\varrho$ (assume that $x^k\in\mathcal{X}_{i^k_t,\ell^k}[j_t^k]$). Then we have                                   
\begin{align}\nonumber                               
\begin{split}
& P\{\left[\left[\varphi\right]\right](s_{\rho}({\bm\cdot};x^0,\chi_u), 0)\ge 0~\vert~\left[\left[\varphi_{\hat{\delta}^{\ast}}\right]\right](s_{\rho^{\ast}}({\bm\cdot};x^{\ast}_{0,{\ell}},u),0)\ge 0, \\
&\forall \ell\}\\                              
&\ge \textstyle\prod\limits_{i=1}^{N^q} P\{\sup_{0\leq t[j_t]\leq t[j_t^{q^i,1}]}\phi_{q^{i}}(\xi^{\ast}_{q^i,\ell_0}(t[j_t];x^0_{q^i},u_{q^i,\ell_0}),   \\&\xi_{q^i,\ell_0}(t[j_t];x^0_{q^i},u_{q^i,\ell_0}))<\hat{\gamma}, \sup_{t[j_t^{q^i,1}]\leq t[j_t]\leq t[j_t^{q^i,2}]}\\&\phi_{q^{i}}(\xi^{\ast}_{q^i,\ell^1}(t-t[j_t^{q^i,1}];x^1,R^{t[i_t^1]}u_{q^i,\ell^1}),\\&      \xi_{q^i,\ell_1}(t-t[j_t^{q^i,1}];x^1,R^{t[i_t^1]}u_{q^i,\ell^1}))<\hat{\gamma},\dots,\\&\sup_{t[j_t^{q^i,m_{q^i}}]\leq t[j_t]\leq T_{\textrm{end}}}\phi_{q^{i}}(\xi^{\ast}_{q^i,\ell^{m_{q^i}}}(t-t[j_t^{q^i,m_{q^i}}];x^{m_{q^i}},R^{t[i_t^{m_{q^i}}]}u_{q^i,\ell^{m_{q^i}}}),\\&\xi_{q^i,\ell^{m_{q^i}}}(t-t[j_t^{q^i,m_{q^i}}];
x^{m_{q^i}},R^{t[i_t^{m_{q^i}}]}u_{q^i,\ell^{m_{q^i}}}))<\hat{\gamma},~\vert~                
\left[\left[\varphi_{\hat{\delta}^{\ast}}\right]\right]\\&(s_{\rho^{\ast}}({\bm\cdot};x^{\ast}_{0,{\ell}},u),0)\ge 0, \forall \ell\}\\
\end{split}  
\end{align}
\begin{align}\nonumber 
\begin{split}
&=\textstyle\prod\limits_{i=1}^{N^q}\big(P\{\sup_{0\leq t[j_t]\leq t[j_t^{q^i,1}]}\phi_{q^{i}}(s_{\rho^{\ast}}(t[j_t];x^0_{q^i},u_{q^i,\ell_0}),\\& s_{\rho}(t[j_t];x^0_{q^i},u_{q^i,\ell_0}))<\hat{\gamma}\}\times\dots \times P\{\sup_{t[j_t^{q^i,m_{q^i}}]\leq t[j_t]\leq T_{\textrm{end}}}\\&\phi_{q^{i}}(s_{\rho^{\ast}}(t[j_t]-t[j_t^{q^i,m_{q^i}}];x^{m_{q^i}}, R^{t[i_t^{m_{q^i}}]}u_{q^i,\ell^{m_{q^i}}}),\\&~~~~~~s_{\rho}(t[j_t]-t[j_t^{q^i,m_{q^i}}];x^{m_{q^i}},R^{t[i_t^{m_{q^i}}]}u_{q^i,\ell^{m_{q^i}}}))<\hat{\gamma}\}\big)\\&> \textstyle\prod\limits_{i=1}^{N^q}\big((1-\frac{\alpha t[j_t^{q^i,1}]}{\hat{\gamma}})\times(1-\frac{\alpha (t[j_t^{q^i,2}]-t[j_t^{q^i,1}])}{\hat{\gamma}})\times\dots
\\&\times(1-\frac{\alpha (T^{q^i}-t[j_t^{q^i,m_{q^i}}])}{\hat{\gamma}})\big)\ge 1-\frac{\alpha T_{\textrm{end}}}{\hat{\gamma}}=1-\epsilon.                                                                                 
\end{split}  
\end{align}

\bibliographystyle{IEEEtran}
\bibliography{zhepowerref}
\end{document}